\documentclass[final,5p,times,twocolumn]{elsarticle}
\usepackage{amssymb}
\usepackage{amsmath}
\usepackage{mathrsfs}
\usepackage{color}
\usepackage{enumerate}
\usepackage[hidelinks]{hyperref} % Colors for links, text and headings

\newtheorem{theorem}{Theorem}
\newproof{proof}{Proof}
\newtheorem{remark}{Remark}
\newtheorem{assumption}{Assumption}
\newtheorem{proposition}{Proposition}

\usepackage{lineno,hyperref}
\modulolinenumbers[5]

\journal{ISA Transactions}

\begin{document}

\begin{frontmatter}
	


\title{Adaptive Super-twisting Observer for Fault Reconstruction \\in Electro-hydraulic Systems}



% Group authors per affiliation:
%\author{Mohamad Bahrami}
%\address{Mechanical Engineering Department, Amirkabir University of Technology, Tehran, Iran}

% or include affiliations in footnotes:
\author[mymainaddress]{Mohamad Bahrami\corref{mycorrespondingauthor}}
\cortext[mycorrespondingauthor]{Corresponding author}
\ead{mohamad.bahrami@aut.ac.ir}
\author[mymainaddress]{Mahyar Naraghi}
\ead{naraghi@aut.ac.ir}
\author[mysecondaryaddress]{Mohammad Zareinejad}
\ead{mzare@aut.ac.ir}
\address[mymainaddress]{Mechanical Engineering Department, Amirkabir University of Technology, Tehran, Iran}
\address[mysecondaryaddress]{New Technologies Research Center, Amirkabir University of Technology, Tehran, Iran}

\begin{abstract}
An adaptive-gain super-twisting sliding mode observer is proposed for fault reconstruction in electro-hydraulic servo systems (EHSS) receiving bounded perturbations with unknown bounds. The objective is to address challenging problems in classic sliding mode observers: chattering effect, conservatism of observer gains, strong condition on the distribution of faults and uncertainties. In this paper, the proposed super-twisting sliding mode observer relaxes the condition on the distribution of uncertainties and faults, and the gain adaptation law leads to eliminate observer gain overestimation and attenuate chattering effects. After using the equivalent output-error-injection feature of sliding mode techniques, a fault reconstruction strategy is proposed. The experimental results are presented, confirming the effectiveness of the proposed adaptive super-twisting observer for precise fault reconstruction in electro-hydraulic servo systems.
\end{abstract}


\end{frontmatter}

%\linenumbers

\section{Introduction}
\label{s:1}
Hydraulic systems are extensively used in industrial fields due to inherent advantages in power transmission through a pressurized fluid \cite{merritt1967hydraulic}. Their industrial applications include active suspension and force control \cite{rajamani1995adaptive,alleyne2000simplified,baghestan2014robust}, positioning \cite{yao2000adaptive,baghestan2015energy,tivay2014switched}, machine tools and manufacturing \cite{chiang2005integrated}, excavating \cite{chiang2004experimental} and flight control \cite{karpenko2009hardware}. The widespread applications and the importance of reliability and safety of hydraulic systems make the fault detection and diagnosis (FDD) an interesting field for control engineers.

The fault detection and diagnosis of electro-hydraulic servo systems (EHSS) is a generally challenging problem because of its highly uncertain nonlinear nature. This nonlinearity includes the dead-zone and hysteresis of control valves and the turbulent fluid flow equations governing the behavior of the overall system. Model uncertainties including parametric uncertainties and uncertain nonlinearities, matched/unmatched disturbances, and friction are among other types of obstacles in achieving a precise fault detection and diagnosis for hydraulic systems. Consequently, a suitable FDD algorithm needs to be proposed that takes into account model uncertainties and disturbances.

Certain efforts have been made in the literature in order to address the FDD of EHSS. These approaches include signal-based and model-based strategies \cite{gao2015survey}. Signal-based approaches including machine learning algorithms \cite{sharifi2018leakage}, the wavelet transformation \cite{goharrizi2010wavelet,goharrizi2011wavelet} and the Hilbert-Huang transformation \cite{goharrizi2012internal} have been studied to detect internal/external leakage. However, these strategies are not applicable in closed-loop tracking, mainly because closed-loop control creates correlation between plant inputs and outputs. Furthermore, signal-based approaches are dependent upon the plant receiving a specific type of input, which is not the case for a plant inside a control loop.
Among model-based approaches, FDD using the adaptive threshold \cite{shi2005development}, unknown input observer \cite{chadli2009state}, Extended Kalman Filter \cite{an2005hydraulic}, adaptive and robust observer \cite{rajamani1995adaptive,hammouri2002failure} and parameter estimation \cite{mahulkar2011derivative} methods have been studied. Nevertheless, these methods have the common disadvantage of sensitivity to unmodeled dynamics. 

Over the past two decades, sliding mode observers for fault reconstruction based on the concept of the so-called equivalent output injection have been proposed \cite{edwards2000sliding,yan2007nonlinear,yan2008robust}. In these studies, a linear/nonlinear system is transformed into a new form with two separate subsystems including the dynamics of unmeasurable and measurable states. After this transformation, a reduced-order Luenberger observer is designed for the first subsystem (i.e. unmeasurable states) and sliding mode observer is proposed for observation of measurable states. Then, system faults associated with the second subsystem are reconstructed through the equivalent output injection. Nevertheless, in many practical applications this strategy suffers from the following:\\
\begin{itemize}
	\item Chattering is a common problem in standard sliding mode observers/controllers that needs to be addressed.
	\item Bounds of the system uncertainties and faults need to be known for observer design.
	\item This method has a challenging problem to reconstruct faults on the first part i.e. unmeasurable states. This is due to the fact that the mentioned method requires a strong condition on the distribution of system uncertainties and faults to hold.
\end{itemize}

This paper addresses these challenging problems. Using second-order sliding mode techniques is a well-known strategy to counteract the chattering effect. Nevertheless, these methods require the time derivative of the sliding variable for their realization. In contrast, the well-known super-twisting algorithm \cite{levant1993sliding} can be realized using only the sliding variable itself. In order to establish the sliding motion, there is no need to know perturbation/fault bounds in the observer design. The gains of the super-twisting algorithm are chosen only in accordance with bounds on the gradient of the perturbation. In practical applications, this bound cannot be effortlessly estimated. As a result, the overestimation of perturbation bounds imposes a conservative choice for super-twisting observer gains and exacerbates the chattering. 

This paper presents a novel adaptive super-twisting (ASTW) observer for fault reconstruction (i.e. internal and external leakages), which takes into account the uncertainties and nonlinearities of the EHSS. Adaptive gains handle the perturbed EHSS with additive perturbations (uncertainties and faults) in which the perturbation bounds are unknown but bounded. These gains dynamically increase until the system states reach sliding motion and then start to reduce towards lower values. The mentioned procedure is repeated whenever the sliding variable or its derivative start to deviate from the sliding manifold. This strategy eliminates the gain overestimation as well as chattering. Furthermore, simultaneous state estimation and fault reconstruction for the EHSS is based on an extended form of previous works \cite{edwards2000sliding,yan2007nonlinear,yan2008robust}. This consideration results in the constraint on the distribution of the perturbation of unmeasurable states being relaxed. In this case, both matched and unmatched disturbances/faults on the EHSS mechanical part can be reconstructed. Finally, the boundedness of adaptive gains has been proven and the finite convergence time is estimated. The stability proof is motivated by recently proposed Lyapunov function \cite{moreno2008lyapunov,shtessel2012novel}.

The rest of this paper is organized as follows: In Section \ref{s:2}, a detailed nonlinear mathematical model of the presented EHSS is described. Design of adaptive super-twisting observer (ASTW), the finite-time stability proof and the proof of the adaptive gains boundedness are presented in Section \ref{s:3}. Section \ref{s:4} gives the fault reconstruction strategy. In Section \ref{s:5}, the experimental set-up and the implementation approach are described. The experimental results are also discussed and the efficacy of the proposed strategy is confirmed.

\section{System modeling and problem statement}
\label{s:2}
Consider the nonlinear dynamic model of the EHSS shown in Fig. \ref{EHSS-scheme}. It is composed of a fixed-displacement hydraulic pump, a proportional relief valve (PRV), a proportional directional valve (PDV), and a double acting cylinder. The mathematical modeling of this system is presented as follows:

\begin{figure}[h]
	\centering
	\includegraphics[width=0.35\textwidth]{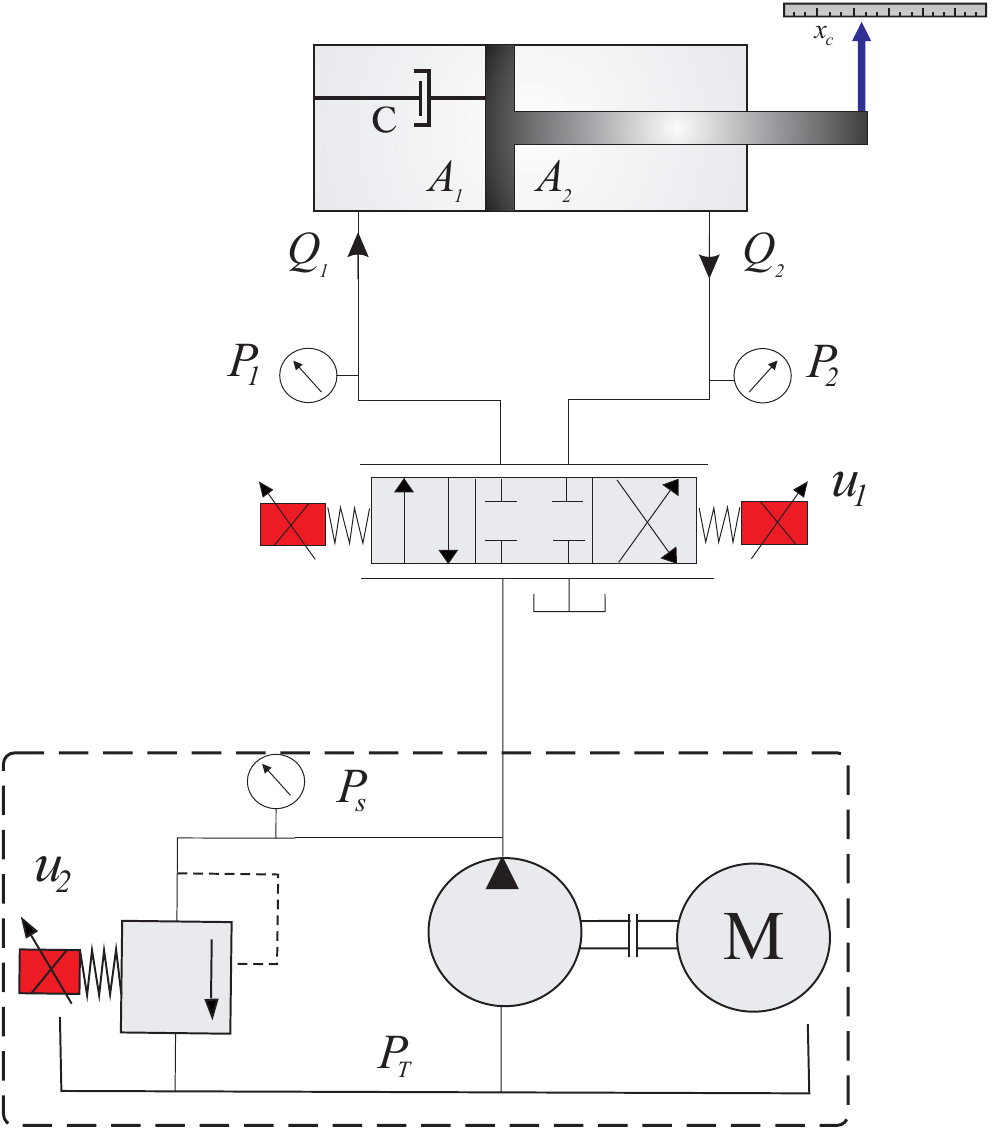}
	\caption{Schematic diagram of the hydraulic system.}\label{EHSS-scheme}
\end{figure}

The proportional directional valve (PDV) model can be described as a first-order system given by:
\begin{equation}\label{eq:valveDyn}
\tau_v\dot{x}_v=-{x}_v+K_v u_1
\end{equation}
where $ x_v $ denotes the spool position, $\tau_v$ is the spool time constant, $ K_v $ and $ u_1 $ are the valve gain and the input voltage of the PDV, respectively.

The pressure dynamics of actuator chambers can be derived as follows: \cite{merritt1967hydraulic}:
\begin{eqnarray}\label{Valvepressure}
\dot{P}_1 = \frac{\beta}{V_{01}+A_1x_c}(+Q_1-A_1\dot{x}_c+Q_{L1}) \\
\dot{P}_2 = \frac{\beta}{V_{02}-A_2x_c}(-Q_2+A_2\dot{x}_c+Q_{L2})
\end{eqnarray}
in which $ \beta $ is the fluid effective bulk modulus, $ V_{01}+A_1x_c $ and $ V_{02}-A_2x_c $ are the volumes of actuator chambers, $ x_c $ and $ \dot{x}_c $ are the position and velocity of the actuator. Terms of $ Q_1 $ and $ Q_2 $ represent flow through PDV orifices, $ Q_{L1} $ and $ Q_{L2} $ represent actuator leakages which can be formulated as follows:
\begin{eqnarray}\label{eq:Valveflow}
& Q_1 =
\begin{cases}
C_d w x_v \sqrt{\frac{2}{\rho}(P_s-P_1)} &  x_v \geq 0\\
C_d w x_v \sqrt{\frac{2}{\rho}(P_1-P_T)} & x_v<0
\end{cases} \\
& Q_2=
\begin{cases}
C_d w x_v \sqrt{\frac{2}{\rho}(P_2-P_T)} &  x_v \geq 0 \\
C_d w x_v \sqrt{\frac{2}{\rho}(P_s-P_2)} & x_v<0
\end{cases} \\ 
& Q_{L1} = +C_{i}(P_2-P_1)-C_{e1}(P_1-P_T) \\
& Q_{L2} = -C_{i}(P_2-P_1)-C_{e2}(P_2-P_T)
\end{eqnarray}
%%
%%
%\begin{eqnarray}\label{Leakflow}
%Q_{L1} = +C_{i}(P_2-P_1)-C_{e1}(P_1-P_T) \\
%Q_{L2} = -C_{i}(P_2-P_1)-C_{e2}(P_2-P_T)
%\end{eqnarray}
%%
%%
with $ C_d $ the discharge coefficient, $ w $ the valve orifice area gradient, $ \rho $ the fluid density, $ P_s $ the supply pressure, $ P_T $ the tank pressure, $ P_1 $, $ P_2 $ the piston and rod side pressures of the cylinder, respectively, $C_i$ the internal leakage coefficient, $C_{e1}$ and $C_{e2}$ being the external leakage coefficients.

The dynamic model of the actuator(cylinder) can be described as:
\begin{equation}\label{eq:cylinder}
m\ddot{x}_c+c\dot{x}_c=P_1A_1 - P_2A_2 + f_d
\end{equation}
where $ m $ and $ c $ are the equivalent mass and the damping coefficient of the environment respectively. $ f_d $ denotes any unknown-but-bounded disturbance and friction force including stiction and Coulomb.

A proportional relief valve (PRV) is considered to control the supply pressure at the location of the pump in accordance with actuator demands. It is worth noting that the importance of controllable supply pressure for the aim of high accuracy positioning and increased efficiency of hydraulic actuators has been addressed by \cite{baghestan2015energy,bahrami2016energy} and \cite{tivay2013energy}. Neglecting the blowdown of the relief valve \cite{baghestan2015energy}, the supply pressure $ P_s $ is related to the PRV control input $ u_2 $ by a first-order system given by: 
\begin{equation}\label{eq:reliefPressure}
\tau_s\dot{P}_s=-{P}_s+K_r u_2
\end{equation}
where $\tau_s$ is the time constant of PRV, $ K_r $ and $ u_2 $ are the valve gain and the input voltage of the PRV, respectively.

Considering equations (\ref{eq:valveDyn}) to (\ref{eq:reliefPressure}) and defining state variables $ x_1=x_v $, $ x_2=P_1 $, $ x_3=P_2 $, $ x_4=P_s $, $ x_5=x_c $, $ x_6=\dot{x}_c $, the nonlinear dynamics of the EHSS can be expressed as:
%%
%%
%\noindent A \hfill Z
%%
\begin{equation}\label{eq:nonlinearDyn}
\begin{array}{l}
\vspace{5pt} \dot{x}_1=-\frac{1}{\tau_v}x_1+\frac{K_v}{\tau_v}u_1  \\
\vspace{5pt} \dot{x}_2=\frac{\beta}{V_{01}+A_1x_5}\left({+C_d w x_1 \sqrt{\frac{2}{\rho}(\Delta P_1)}-A_1x_6+Q_{L1}(x_2,x_3)}\right)  \\
\vspace{5pt} \dot{x}_3=\frac{\beta}{V_{02}-A_2x_5}\left({-C_d w x_1 \sqrt{\frac{2}{\rho}(\Delta P_2)}+A_2x_6+Q_{L2}(x_2,x_3)}\right)  \\
\vspace{5pt} \dot{x}_4=-\frac{1}{\tau_s}x_4+\frac{K_r}{\tau_s}u_2+\Delta(x)  \\
\vspace{5pt} \dot{x}_5=x_6  \\
\vspace{5pt} \dot{x}_6=-\frac{c}{m}x_6+\frac{A_1}{m}x_2-\frac{A_2}{m}x_3+\frac{f_d}{m}\\
y = [x_2,x_3,x_4,x_5]^{T}
\end{array}
\end{equation}
in which $ y $ denotes measurable states, and pressure differences depend on the position of the PDV spool:\\
For 
$ x_1 \geq 0  $ :
$ \begin{cases}
\Delta P_1 = x_4-x_2 \\
\Delta P_2 =  x_3-P_T
\end{cases} $
\\and\\ for
$ x_1 < 0  $ :
$  \begin{cases}
\Delta P_1 = x_2-P_T  \\
\Delta P_2 = x_4-x_3
\end{cases} $
\\
The unknown-but-bounded term $ \Delta(x) $ denotes the unmodeled dynamics and uncertain nonlinearities of supply pressure dynamics.

Given the nonlinear dynamics of EHSS (\ref{eq:nonlinearDyn}) with a variable structure in the pressure dynamics of cylinder chambers, the main sources of system nonidealities include leakage faults due to breakdown of seals, external disturbances on the cylinder, and uncertainties in supply pressure dynamics. Considering these aspects, it is evident that a suitable observer needs to be proposed that is robust to mentioned system nonidealities.
The development of adaptive super-twisting observer will be described in the following section.
\section{Adaptive super-twisting observer design}
\label{s:3}
The nonlinear system (\ref{eq:nonlinearDyn}) includes the unavailable state ($ x_1 $), measurable states ($ \vec{y} $) and a second-order system (cylinder dynamics). Our objective is to reorder the states so that the system is turned into a suitable form for super-twisting observer design. With this in mind, the system (\ref{eq:nonlinearDyn}) can be rewritten as:
\begin{equation}\label{eq:nonlinearDyn2}
\begin{array}{l}
\hspace{12pt} \vspace{5pt}\dot{z}_1=-\frac{1}{\tau_v}z_1+\frac{K_v}{\tau_v}u_1 \\
\vspace{5pt} \left\{\begin{array}{l}\vspace{5pt}
\dot{y}_1= \phi_1(y)\left({+Q_1(z,y)-A_1z_2}\right)+\phi_1(y)f_1 \\
\vspace{5pt} \dot{y}_2=\phi_2(y)\left({-Q_2(z,y)+A_2z_2}\right)+\phi_2(y)f_2 \\
\vspace{0pt} \dot{y}_3=-\frac{1}{\tau_s}y_3+\frac{K_r}{\tau_s}u_2 +\Delta(x)\\
\end{array}\right.\\
\left\{\begin{array}{l}
\vspace{5pt} \dot{y}_4=z_2  \\
\vspace{0pt} \dot{z}_2=-\frac{c}{m}z_2+\frac{A_1}{m}y_1-\frac{A_2}{m}y_2+\frac{f_d}{m}
\end{array}\right.
\end{array}
\end{equation}
where $ z=col(z_1=x_1,z_2=x_6) $ denotes unavailable states and $ y=col(y_1,y_2,y_3,y_4) $ denotes measurable states, $ \phi_1 = \frac{\beta}{V_{01}+A_1y_4}  $ , $ \phi_2 = \frac{\beta}{V_{02}-A_2y_4}  $, $ f_1=Q_{L1}(y_1,y_2) < f_{1_{max}}$ and $ f_2=Q_{L2}(y_1,y_2) < f_{2_{max}}$ are system faults (leakages in cylinder chambers).

The ASTW observer is represented by the following dynamical system
\begin{equation}\label{eq:stwObs}
\begin{array}{l}
\vspace{5pt} \dot{\hat{z}}_1=-\frac{1}{\tau_v}\hat{z}_1+\frac{K_v}{\tau_v}u_1 \\
\vspace{5pt} \dot{\hat{y}}_1=\phi_1(y)\left({+Q_1(\hat{z},y)-A_1\hat{z}_2}\right) + \mu(y_1,\hat{y}_1)\\
\vspace{5pt} \dot{\hat{y}}_2=\phi_2(y)\left({-Q_2(\hat{z},y)+A_2\hat{z}_2}\right) + \mu(y_2,\hat{y}_2)\\
\vspace{5pt} \dot{\hat{y}}_3=-\frac{1}{\tau_s}y_3+\frac{K_r}{\tau_s}u_2 + \mu(y_3,\hat{y}_3) \\
\vspace{5pt} \dot{\hat{y}}_4=\hat{z}_2 +L_1(t) \mu_1(y_4,\hat{y}_4)  \\
\vspace{5pt} \dot{\hat{z}}_2=-\frac{c}{m}\hat{z}_2+\frac{A_1}{m}y_1-\frac{A_2}{m}y_2+\frac{f_d}{m} +L_2(t) \mu_2(y_4,\hat{y}_4)
\end{array}
\end{equation}
\begin{equation}
\hat{y}=[\hat{x}_2,\hat{x}_3,\hat{x}_4,\hat{x}_5]^T
\end{equation}
in which $ \mu(\cdot) $ is the so-called output error-injection term such that:
\begin{equation}\label{eq:output_inj}
\left\{
\begin{array}{l}
\mu(\sigma) \hspace{4pt}=L_1(t) \mu_1(\sigma)+ \int_{0}^{t} L_2(t)\mu_2(\sigma) d\tau  \\
\mu_1(\sigma) = |\sigma|^{\frac{1}{2}}sign(\sigma)  \\
\mu_2(\sigma) = sign(\sigma)
\end{array}
\right.
\end{equation}
where $ \sigma $ denotes the sliding variable and adaptive gains $ L_1(t) $ and $ L_2(t) $ are to be determined.  
\begin{remark}
			Consider Eq. (\ref{eq:output_inj}) as the observer effort, it generates the continuous signal which drives the sliding variable and its derivative to zero in finite time in the presence of unknown-but-bounded perturbations. Since term $ \mu(\sigma) $ contains a continuous function $ L_1(t) |\sigma|^{\frac{1}{2}}sign(\sigma) $ and a discontinuous function under the integral $ L_2(t) \int_{0}^{t} sign(\sigma) d\tau $, chattering is attenuated but is not eliminated. Furthermore, the adaptation of observer gains ($ L_1(t) $ and $ L_2(t) $) improves the chattering reduction because it eliminates the conservatism of observer gains.
\end{remark}

It is worth noting that the presented system in form (\ref{eq:nonlinearDyn2}) is an extension of previous works \cite{yan2007nonlinear,yan2008robust}. In mentioned studies, the system falls into two separate subsystems including the dynamics of measurable and unmeasurable states that each subsystem contains a set of first-order systems. Finally, a first-order sliding mode observer is used for state observation and fault reconstruction. In proposed form (\ref{eq:nonlinearDyn2}), the second-order system (cylinder dynamics) is considered as a separate subsystem and using ASTW observer (\ref{eq:stwObs}) provides the reconstruction of disturbances/faults on the unmeasurable part (cylinder velocity). Consequently, there is no need for any conditions on the distribution of system faults. Although in previous studies, this requirement is  necessary. 
\begin{remark}
	The presented form (\ref{eq:nonlinearDyn2}-\ref{eq:stwObs}) can be generalized for a class of nonlinear systems for the objective of nonlinear observer design and/or fault reconstruction. Some of these nonlinear systems are robot manipulators powered by electrical, hydraulic or pneumatic drive systems, turbocharged Engines, PEM fuel cell air-feed systems, switching power converters, etc. 
	This generalization will be studied and presented in future work.
\end{remark}

Let the observation errors be: $ e_{z_1}=z_1-\hat{z}_1 $, $ e_{y_1}=y_1-\hat{y}_1 $, $ e_{y_2}=y_2-\hat{y}_2 $, $ e_{y_3}=y_3-\hat{y}_3 $, $ e_{y_4}=y_4-\hat{y}_4 $, $ e_{z_2}=z_2-\hat{y}_2 $. Subtracting (\ref{eq:stwObs}) from (\ref{eq:nonlinearDyn2}), the observation error dynamical equation can be described by:
\begin{align}\label{eq:obsE1}
& \dot{e}_{z_1}=-\frac{1}{\tau_v}e_{z_1} \\ \label{eq:obsY1}
& \dot{e}_{y_1}=\tilde{G}_1+\phi_1(y)f_1-\mu(y_1,\hat{y}_1) \\ \label{eq:obsY2}
& \dot{e}_{y_2}=\tilde{G}_2+\phi_2(y)f_2- \mu(y_2,\hat{y}_2) \\ \label{eq:obsY3}
& \dot{e}_{y_3}=- \mu(y_3,\hat{y}_3)+\Delta(x) \\ \label{eq:obsY4}
& \dot{e}_{y_4}=-L_1(t) \mu_1(y_4,\hat{y}_4)+e_{z_2}  \\ \label{eq:obsV}
& \dot{e}_{z_2} =-L_2(t)\mu_2(y_4,\hat{y}_4)+\varrho_4(e_{z_2},t)
\end{align}
where $ \varrho_4(e_{z_2},t) =-\frac{c}{m}e_{z_2}+\frac{f_d}{m} $,\\
$ \tilde{G}_1(z,\hat{z},y)=\phi_1(y)\left({+Q_1(z,y)-Q_1(\hat{z},y)-A_1e_{z_2}}\right) $ and\\
$ \tilde{G}_2(z,\hat{z},y)=\phi_2(y)\left({-Q_2(z,y)+Q_2(\hat{z},y)+A_2e_{z_2}}\right) $.

The stability of dynamic error systems (\ref{eq:obsE1}-\ref{eq:obsV}) can be split into two parts including the stability of unmeasurable part (\ref{eq:obsE1}) and the other parts (\ref{eq:obsY1}-\ref{eq:obsV}).\\
In the first step, the stability of the unmeasurable part (\ref{eq:obsE1}) is considered. It is obvious that the scaler system (\ref{eq:obsE1}) is isolated from the other parts and is exponentially stable. Consequently, $ \lim_{t\rightarrow\infty} e_{z_1}(t) = 0 $ and $ e_{z_1},\dot{e}_{z_1} \in \mathcal{L}^{\infty} $

In the second step, the stability of the other parts is performed as follows. Bearing in mind that the cylinder position, the PDV spool, and fluid flow passing through the PDV are physically bounded. It can be easily shown that the nonlinear terms $ \tilde{G}_1 $ and $ \tilde{G}_2 $ are globally bounded by
\begin{align}\label{eq:lips}
& \lvert\tilde{G}_1(z,\hat{z},y)\rvert \leq C_d w\sqrt{\frac{2}{\rho}}\bar{\phi}_1 \lvert e_{z_1} \sqrt{\Delta P_1} \rvert  +A_1\bar{\phi}_1\lvert e_{z_2} \rvert = |{^1 \varrho_1}(z,\hat{z},y)| \leq {^1 \delta_1}  \nonumber \\
& \lvert\tilde{G}_2(z,\hat{z},y)\rvert \leq C_d w\sqrt{\frac{2}{\rho}}\bar{\phi}_2 \lvert e_{z_1} \sqrt{\Delta P_2} \rvert +A_2\bar{\phi}_2\lvert e_{z_2} \rvert= |{^1 \varrho_2}(z,\hat{z},y)| \leq {^1 \delta_2}
\end{align}
where $ \bar{\phi}_1 $ and $ \bar{\phi}_2 $ are upper bounds on $ \phi_1(y) $ and $ \phi_2(y) $, respectively. $ {^1 \delta_1} $ and $ {^1 \delta_2} $ are unknown positive constants\footnote{The notations $ {^1 \delta_i} $ and $ {^2 \delta_i} $ indicate the bound $ ^1 \delta $ on the $ i $th nonlinear term and the bound $ ^2 \delta $ on the time-derivative of the $ i $th nonlinear term, respectively}. The reader refer to \cite{moreno2008lyapunov} for more information about necessary conditions on perturbation bounds in super-twisting algorithm.

\begin{assumption}\label{assum1}
	The perturbation $ \varrho_4(e_{z_2},t) $, uncertainty $ \Delta(x) $, and the time derivatives of nonlinear terms $ \phi_1(y)f_1 $ and $ \phi_2(y)f_2 $ , in the error dynamics (\ref{eq:obsY1}-\ref{eq:obsV}) are bounded.
	\begin{align}\label{eq:bound}
	& |\varrho_4(e_{z_2},t) | \leq {^2 \delta_4} \nonumber \\
	& | \Delta(x) | \leq {^1 \delta_3} \nonumber \\
	& \lvert\frac{d}{dt}\left(\phi_1(y)f_1\right)\rvert =|{^2 \varrho_1}(e_{z_2},y)| \leq {^2 \delta_1} \nonumber \\
	& |\frac{d}{dt}\left(\phi_2(y)f_2\right)|=|{^2 \varrho_2}(e_{z_2},y)| \leq {^2 \delta_2} 
	\end{align}
	where $ {^2 \delta_1},{^2 \delta_2},{^2 \delta_4}, {^1 \delta_3} $ are unknown positive constants.
\end{assumption}
\begin{remark}
	The nonlinear system (\ref{eq:nonlinearDyn2}) and the ASTW observer (\ref{eq:stwObs}) are bounded-input bounded-state (BIBS) in finite time, since the system (\ref{eq:nonlinearDyn2}) is a physical system \cite{perruquetti2002sliding} and it is well-known in literature that hydraulic systems have bounded state for bounded input \cite{merritt1967hydraulic}. With this in mind, it is reasonable to assume that, all of the forces acting on the cylinder are bounded. Therefore, the cylinder acceleration and velocity are bounded and the first condition of Assumption \ref{assum1} hold. Notice that system faults (internal and external leakages) and their derivatives are physically bounded. Therefore, the last two conditions of Assumption \ref{assum1} also hold on the overall region of the system behavior.  
\end{remark}
For the system (\ref{eq:obsY1}-\ref{eq:obsV}), consider the following sliding surface:
\begin{equation}\label{eq:slidingsurf}
\boldsymbol{\sigma} = \{col(e_{z_1},e_{z_2},e_{y_1},e_{y_2},e_{y_3},e_{y_4})|e_{y_i}=0 , i=1:4\}
\end{equation}
\begin{theorem}\label{Theorem1}
	Consider the dynamic of observation error (\ref{eq:obsY1}-\ref{eq:obsV}) and suppose that (\ref{eq:lips}) and (\ref{eq:bound}) hold, Then, for any initial conditions $ e_{y_i}(0),e_{z_2}(0),\boldsymbol{\sigma(0)} $, 
	if the adaptive gain $ L_1(t) $ satisfies the following condition\footnote{The index $ i $ is associated with bounds in (\ref{eq:bound}). This index means that, there exist 4 bounds (Eq. \ref{eq:condition}) and 4 adaption laws (Eqs. \ref{eq:adapgain}) for each subsystems in the dynamic of observation error (\ref{eq:obsY1}-\ref{eq:obsV}).}
	\begin{equation}
	\begin{array}{l}\label{eq:condition}
	L_1 > \frac{\left[2\lambda_1(\lambda_1-{^1 \delta_i})+(\lambda_2+{^2 \delta_i})\right]^2}{4\lambda_1\lambda_2}+\\
	\hspace{100pt}\frac{4\lambda_1(2\lambda_1-{^2 \delta_i})+\lambda_2{^1 \delta_i}}{2\lambda_2} , i=1:4
	\end{array}
	\end{equation}
	there exist a finite time $ t_f>0 $ and a parameter $ \epsilon_i $, so that the trajectories of the observation error system (\ref{eq:obsY1}-\ref{eq:obsV}) can be driven to the sliding surface (\ref{eq:slidingsurf}) with the adaptation law given by:
	\begin{align}\label{eq:adapgain}
	&\dot{L}_1 =
	\begin{cases}
	\alpha_1\sqrt{\frac{\varGamma_1}{2}}sign(|\sigma|-\epsilon_i) & L_1 > \underline{L}  \nonumber \\ 
	\bar{L}       &  L_1 \leq \underline{L}
	\end{cases} ,\quad  (L_1(0) > \underline{L})\\
	&  L_2=\lambda_1L_1
	\end{align}
	in which $ \alpha_1, \varGamma_1, \epsilon_i, \lambda_1, \lambda_2, \bar{L}, \underline{L} $ are arbitrary positive constants, among which $ \underline{L} $ can be arbitrarily small.
\end{theorem}

\begin{proof}
	Consider the system (\ref{eq:obsY1}), it can be rewritten as:
	\begin{align}\label{eq:STWdyn}
	& \dot{e}_{y_1}=-L_1(t) |e_{y_1}|^{\frac{1}{2}}sign(e_{y_1})+\chi + {^1 \varrho_1} \nonumber \\  & \dot{\chi}=-L_2(t)sign(e_{y_1})+{^2 \varrho_1}
	\end{align}
	where $ |{^1 \varrho_1}| \leq {^1 \delta_1} $ and $ |{^2 \varrho_1}| \leq {^2 \delta_1} $ under conditions (\ref{eq:lips}) and (\ref{eq:bound}). \\
	It can be observed that Eq. (\ref{eq:STWdyn}) has the same form as the second-order system described by Eqs. (\ref{eq:obsY4}-\ref{eq:obsV}). Furthermore, Eqs. (\ref{eq:obsY2})  and (\ref{eq:obsY3}) can also be rewritten to this form. Therefore, we only focus on the stability proof of Eq. (\ref{eq:obsY1}). The same procedure can be followed for systems (\ref{eq:obsY2}), (\ref{eq:obsY3}) and (\ref{eq:obsY4}-\ref{eq:obsV}).
	In order to conduct a more convenient Lyapunov analysis, a new state vector is introduced as

	\begin{equation}\label{eq:newstate}
	\zeta=
	\begin{bmatrix}
	\zeta_1  &
	\zeta_2
	\end{bmatrix}^T =
	\begin{bmatrix}
	|e_{y_1}|^{\frac{1}{2}}sign(e_{y_1})  &
	\chi
	\end{bmatrix}^T
	\end{equation}
	Then, the system (\ref{eq:STWdyn}) can be rewritten as:
	\begin{equation}\label{eq:newdynsys}
	\begin{array}{l}
	\vspace{3pt}\dot{\zeta}= \frac{1}{|\zeta_1|}
	\left[
	\begin{array}{cc}
	-\frac{L_1}{2} & \frac{1}{2}  \\
	-L_2 & 0
	\end{array}
	\right] \zeta
	+
	\frac{1}{|\zeta_1|}
	\left[
	\begin{array}{c}
	{^1 \varrho_1} \\
	{^2 \varrho_1}|\zeta_1| \\
	\end{array}
	\right]
	=
	\frac{1}{|\zeta_1|}
	\underbrace{\left[
		\begin{array}{cc}
		-\frac{L_1}{2}+{^1 \delta_1} & \frac{1}{2}  \\
		-L_2+{^2 \delta_1} & 0
		\end{array}
		\right]}_\text{$ \Lambda $}
	\zeta
	\end{array}
	\end{equation}
	it can be observed that $ \sigma = e_{y_1},\dot{\sigma}=\dot{e}_{y_1} \rightarrow 0 $ in finite time if $ \zeta_1,\zeta_2 \rightarrow 0 $ in finite time. To this end, the following Lyapunov function candidate is introduced:
	\begin{equation}\label{eq:lyap}
	V=V_1+\frac{1}{2\varGamma_1}\left(L_1-L^*_1\right)^2+\frac{1}{2\varGamma_2}\left(L_2-L^*_2\right)^2
	\end{equation}
	where
	\begin{equation}
	V_1=\zeta^TP\zeta ,\hspace{5pt} P= \begin{bmatrix}
	4\lambda_1^2+2\lambda_2  & -2\lambda_1 \\
	-2\lambda_1 & 1
	\end{bmatrix}
	\end{equation}
	and $ L^*_1 $, $ L^*_2 $ are positive constants that exist so that $ L_1-L^*_1 \leq 0 $ and $ L_2-L^*_2 \leq 0 $, $ \forall t \geq 0 $. It is necessary to note that the adaptive gains $ L_1 $ and $ L_2 $ are bounded (its proof will be given later). It can be observed that the matrix $ P $ is symmetric and positive definite if $ \lambda_1 $, $ \lambda_2 $ have any real and positive value.\\
	The time derivative of $ V $ along the system (\ref{eq:newdynsys}) is given by:
	\begin{equation}\label{eq:vdot}
	\dot{V}=\frac{1}{|\zeta_1|}\zeta^T\underbrace{\left(P\Lambda+\Lambda^TP\right)}_\text{$ -\Omega $}\zeta+\frac{1}{\varGamma_1}\tilde{L}_1\dot{L}_1+\frac{1}{\varGamma_2}\tilde{L}_2\dot{L}_2
	\end{equation}
	in which $ \tilde{L}_1 = L_1-L^*_1 $, $ \tilde{L}_2 = L_2-L^*_2 $ and\\
	\begin{equation}
	\Omega=\begin{bmatrix}
	2(4\lambda_1^2+2\lambda_2)(\frac{L_1}{2}-{^1 \delta_1})-4\lambda_1(L_2-{^2 \delta_1})  & \star \\
	L_2-\lambda_2-{^2 \delta_1}-2\lambda_1(\frac{L_1}{2}-{^1 \delta_1})-2\lambda_1^2 & 2\lambda_1
	\end{bmatrix}
	\end{equation}
	where the symbol $ \star $ indicates a symmetric element.\\
	The first term of Eq. (\ref{eq:vdot}) (which represents $ \dot{V}_1 $) is negative definite if $ \Omega = \Omega^T >0 $. It is easy to show that the matrix $ \Omega $ will be positive definite if the gains $ L_1 $ and $ L_2 $ satisfy the following conditions
	\begin{flalign}\label{eq:cond1}
	& L_2 = \lambda_1L_1 \\ \label{eq:cond2}
	& L_1 > \frac{\left[2\lambda_1(\lambda_1-{^1 \delta_1})+(\lambda_2+{^2 \delta_i})\right]^2}{4\lambda_1\lambda_2}+\frac{4\lambda_1(2\lambda_1-{^2 \delta_1})+\lambda_2{^1 \delta_1}}{2\lambda_2}
	\end{flalign}
	assuming that (\ref{eq:cond1}) and (\ref{eq:cond2}) hold, it implies
	\begin{equation}\label{eq:vdot2}
	\dot{V}\leq -\frac{1}{|\zeta_1|}\lambda_{min}\left(\Omega\right)\|\zeta\|_2^2+\frac{1}{\varGamma_1}\tilde{L}_1\dot{L}_1+\frac{1}{\varGamma_2}\tilde{L}_2\dot{L}_2
	\end{equation}
	Considering the following facts
	\begin{align}
	& \lambda_{min}(P)\|\zeta\|_2^2\leq\zeta^TP\zeta\leq\lambda_{max}(P)\|\zeta\|_2^2\\
	& |\zeta_1|=|e_{y_1}|^{\frac{1}{2}}\leq\|\zeta\|_2\leq\frac{V_1^{\frac{1}{2}}}{\lambda_{min}^{\frac{1}{2}}(P)}
	\end{align}
	Then, the inequality (\ref{eq:vdot2}) can be rewritten as
	\begin{equation}\label{eq:39}
	\begin{array}{l}
	\dot{V}\leq -c_1V_1^{\frac{1}{2}}+\frac{1}{\varGamma_1}\tilde{L}_1\dot{L}_1+\frac{1}{\varGamma_2}\tilde{L}_2\dot{L}_2\\
	\hspace{12pt}=\underbrace{-c_1V_1^{\frac{1}{2}}-\frac{\alpha_1}{\sqrt{2\varGamma_1}}|\tilde{L}_1|-\frac{\alpha_2}{\sqrt{2\varGamma_2}}|\tilde{L}_2|}_\Xi+\frac{1}{\varGamma_1}\tilde{L}_1\dot{L}_1\\
	\hspace{23pt}+\frac{\alpha_1}{\sqrt{2\varGamma_1}}|\tilde{L}_1|+\frac{1}{\varGamma_2}\tilde{L}_2\dot{L}_2+\frac{\alpha_2}{\sqrt{2\varGamma_2}}|\tilde{L}_2|
	\end{array}
	\end{equation}
	where $ c_1 = \frac{\lambda_{min}^{\frac{1}{2}}(P)}{\lambda_{max}(P)}\lambda_{min}\left(\Omega\right) $, and it is obvious that term $ \Xi $ is negative definite.
	\\Taking into account the well-known inequality $ \left(\sum_{i=1}^3z_i^2\right)^\frac{1}{2} \leq \ \sum_{i=1}^3|z_i| $, and the Lyapunov function (\ref{eq:lyap}), the term $ \Xi $ can be rewritten as:
	\begin{eqnarray}\label{eq:40}
	\Xi = -c_1V_1^{\frac{1}{2}}-\frac{\alpha_1}{\sqrt{2\varGamma_1}}|\tilde{L}_1|-\frac{\alpha_2}{\sqrt{2\varGamma_2}}|\tilde{L}_2|\leq -\gamma\sqrt{V}
	\end{eqnarray}
	with $ \gamma=min\{c_1,\alpha_1,\alpha_2\} $. Substituting (\ref{eq:40}) into Eq. (\ref{eq:39}), yields:
	\begin{equation}\label{eq:41}
	\begin{array}{l}
	\dot{V}\leq -\gamma D^+\sqrt{V}+\frac{1}{\varGamma_1}\tilde{L}_1\dot{L}_1\\
	\hspace{23pt}+\frac{\alpha_1}{\sqrt{2\varGamma_1}}|\tilde{L}_1|+\frac{1}{\varGamma_2}\tilde{L}_2\dot{L}_2+\frac{\alpha_2}{\sqrt{2\varGamma_2}}|\tilde{L}_2|\\
	\hspace{12pt}\vspace{3pt}=-\gamma D^+\sqrt{V}-|\tilde{L}_1|\left(\frac{1}{\varGamma_1}\dot{L}_1-\frac{\alpha_1}{\sqrt{2\varGamma_1}}\right)-|\tilde{L}_2|\left(\frac{1}{\varGamma_2}\dot{L}_2-\frac{\alpha_2}{\sqrt{2\varGamma_2}}\right)
	\end{array}
	\end{equation}
	where $ D^+ $ indicates the upper right-hand derivative of $ \sqrt{V} $ with respect to its argument. This notation is used to denote the positive root of $ \sqrt{V} $ in a more convenient notation \cite{khalil1996noninear}.\\
	Given the proposed adaptation law (\ref{eq:adapgain}), the Eq. (\ref{eq:41}) can be analyzed in two situations $ A1 $ and $ A2 $ where
	\begin{itemize}
		\item[\textbf{A1}.] $ \vspace{3pt} \left\{
		\begin{array}{l}
		|\sigma| > \epsilon  \\
		L_1 > \underline{L}
		\end{array}
		\right. \forall t \geq 0 $. then, in view of (\ref{eq:adapgain})
		\begin{equation}
		\dot{L}_1=\alpha_1\sqrt{\frac{\varGamma_1}{2}}
		\end{equation}
		Considering condition (\ref{eq:cond1}), its derivative yields:
		\begin{equation}
		\dot{L}_2=\alpha_2\sqrt{\frac{\varGamma_2}{2}}
		\end{equation}
		where $ \lambda_2  $ is chosen such that $ \lambda_2=\frac{\alpha_2}{\alpha_1}\sqrt{\frac{\varGamma_2}{\varGamma_1}} $.
		%		 \begin{eqnarray}
		%		 \lambda_2=\frac{\alpha_2}{\alpha_1}\sqrt{\frac{\varGamma_2}{\varGamma_1}}
		%		 \end{eqnarray}
		Then, it can be observed that Eq. (\ref{eq:41}) is reduced to the following:
		\begin{eqnarray}\label{eq:vdotT}
		\dot{V}\leq -\gamma D^+\sqrt{V}
		\end{eqnarray}
		
		\item[\textbf{A2}.] $ |\sigma| < \epsilon $, then  $ L_1(t) $ enters into decreasing situation and (\ref{eq:adapgain}) becomes:
		\begin{equation}\label{eq:45}
		\dot{L}_1 =
		\begin{cases}
		-\alpha_1\sqrt{\frac{\varGamma_1}{2}} & L_1 > \underline{L} \\
		\bar{L}       &  L_1 \leq \underline{L}
		\end{cases}
		\end{equation}
		In view of Eq. (\ref{eq:45}), Eq. (\ref{eq:41}) can be rewritten as:
		\begin{equation}\label{eq:vdot3}
		\dot {V}\leq-\gamma D^+\sqrt{V}+\varepsilon
		\end{equation}
		where
		\begin{equation}\label{eq:47}
		\varepsilon =
		\begin{cases} \vspace{5pt}
		2|\tilde{L}_1|\frac{\alpha_1}{\sqrt{2\varGamma_1}}+2|\tilde{L}_2|\frac{\alpha_2}{\sqrt{2\varGamma_2}} & L_1 > \underline{L} \\
		-|\tilde{L}_1|\left(\frac{\bar{L}}{\varGamma_1}-\frac{\alpha_1}{\sqrt{2\varGamma_1}}\right)\\
		\hspace{30pt}-|\tilde{L}_2|\left(\frac{\lambda_1\bar{L}}{\varGamma_2}-\frac{\alpha_2}{\sqrt{2\varGamma_2}}\right)       &  L_1 \leq \underline{L}
		\end{cases}
		\end{equation}
		It can be observed that $ \varepsilon  $ can become positive. Note that the second equation of (\ref{eq:47}) is valid only for a finite time interval, since if $ L_1 \leq \underline{L} $, its value immediately starts increasing in accordance with $ L_1 = \bar{L}t + \underline{L} $, Then the first equation in (\ref{eq:47}) will be established. According to the description given for (\ref{eq:47}), Eq. (\ref{eq:41}) can become sign indefinite and $ |\sigma(\cdot)| $ may become greater than $ \epsilon $ due to reduce of the ASTW observer gains $ L_1(t) $ and $ L_2(t) $.
	\end{itemize}
	
	As soon as the magnitude of sliding variable $ |\sigma(\cdot)| $ becomes larger than $ \epsilon $ the defined situation $ A1 $ holds so that $ \sigma(\cdot) $ reaches the domain $ |\sigma(\cdot)|<\epsilon $ again in a finite time. Consequently, during the presented adaptation process the sliding variable $ \sigma(\cdot) $ first reaches the domain $ |\sigma(\cdot)|<\epsilon $, then may leave it, both in a finite time. This behavior may occur repeatedly around the sliding surface and it is guaranteed that the sliding variable $ \sigma(\cdot) $ always stays in a greater domain $ |\sigma(\cdot)|<\bar{\epsilon} , \epsilon<\bar{\epsilon} $ in the sliding mode strategy.
	
	Considering the comparison principle and the inequality (\ref{eq:vdotT}), it implies $ \zeta_1,\zeta_2 \rightarrow 0 $ in finite time $ T_r\leq\frac{2\sqrt{V(t_0)}}{\gamma} $ \cite{shtessel2012novel}.\\
	Theorem \ref{Theorem1} is proven.
\end{proof}

Now the stability of ASTW observer (\ref{eq:stwObs}-\ref{eq:output_inj}) is proved. The boundedness of adaptive gains is considered in the following proposition.
\begin{proposition}\label{Prop1}
	Given the adaptation law (\ref{eq:adapgain}), adaptive gains $ L_1(t) $ and $ L_2(t) $ are bounded , $ \forall t \geq 0 $
\end{proposition}
\begin{proof}
	Inside the domain $ \epsilon < |\sigma(\cdot)| \leq \bar{\epsilon} $, a solution to (\ref{eq:adapgain}) can be obtained as
	\begin{equation}\label{eq:boundedg}
	L_1(t)=\alpha_1\sqrt{\frac{\varGamma_1}{2}}t+L_1(0),\quad  0\leq t \leq T_r
	\end{equation}
	Since $ L_2=\lambda_1L_1 $ and in view of Eq. (\ref{eq:boundedg}), adaptive gains $ L_1(t) $ and $ L_2(t) $ are bounded. Inside the domain $ |\sigma(\cdot)| \leq \epsilon $, adaptive gains are decreasing and have a lower bound. Consequently, the observer gains are bounded in all over \emph{the region of attraction} ($ -\bar{\epsilon} < \sigma(\cdot) \leq \bar{\epsilon} $).\\
	Proposition \ref{Prop1} is proven.
\end{proof}
\begin{remark}
	It is necessary to note that, tuning the parameter $ \epsilon $ is an important factor in the performance of the proposed observer. In practical applications, this parameter tuning is affected by the amplitude of Lebesgue-measurable noise that exists in signals. A too small value for $ \epsilon $ lower than the noise amplitude yields system trajectories never reach to a lower value than $ \epsilon $. Therefore, in such situation and in view of (\ref{eq:adapgain}) the observer gains never decrease which induces large oscillation and parameter drift. If parameter $ \epsilon $ is too large, system trajectories reach to a vicinity of the sliding surface around $ \epsilon $. This situation induces a state estimation with low accuracy. Consequently, a suitable value $ \epsilon $ should be chosen associated with the consideration of implementation (i.e. the amplitude of noise and the desired accuracy for estimation). 
\end{remark}
\begin{remark}
	Parameter $ \bar{\epsilon} $ represents a finite region of attraction. In order to establish sliding motion and maintain it in the presence of system perturbations, the region of attraction should be large enough that yields conservatism of observer gains. In classic sliding mode technique, the size of the region of attraction can be estimated. Although in this paper the size of the region of attraction ($ \bar{\epsilon} $) dynamically increases or decreases as well as adaptive gains according to the bound of system perturbations. Therefore the estimation of $ \bar{\epsilon} $ is practically impossible.
\end{remark}
\section{Fault reconstruction}\label{s:4}
In this section, we use the so-called equivalent output injection concept \cite{edwards2000sliding} in order to reconstruct system faults (cylinder leakages). It is assumed that Eqs. (\ref{eq:lips}-\ref{eq:condition}) hold. It is guaranteed that during the sliding motion the sliding variable and its $r-1$ consecutive derivatives are maintained at zero in an $ r $th-order sliding mode. Therfore, when sliding motion is established:

\begin{equation}\label{eq:eq}
e_{y_i}=\dot{e}_{y_i}=0 \quad \text{for}\quad i=1,2,3 \quad \text{and} \quad e_{y_4}=e_{z_2}=0
\end{equation}

in view of Eqs. (\ref{eq:obsY1}) and (\ref{eq:obsY2}), and considering (\ref{eq:eq}) the equivalent output injection can be obtained as:
\begin{equation}\label{eq:faultrec}
0=\tilde{G}_i(y,0,e_{z_1}) +\phi_i(y)f_i-\mu_{eq}(y_i,\hat{y}_i) ,\hspace{5pt} i=1,2
\end{equation}
As mentioned earlier, $ \lim_{t\rightarrow\infty} e_{z_1}(t) = 0 $ and since $ \phi_i(y), i=1,2 $ are nonsingular, then the estimation of faults can be constructed as:
\begin{equation}\label{eq:faultreco}
\hat{f}_i=\phi_i^{-1}(y)\mu_{eq}(y_i,\hat{y}_i) ,\hspace{5pt} i=1,2
\end{equation}
in which $ \mu_{eq}(\cdot,\cdot) $ denotes the equivalent output-injection signal.
\begin{theorem}\label{theorem2}
	Assume that all of the conditions in Theorem \ref{Theorem1} hold. Then, during sliding motion, a precise reconstruction of faults $ f_1 $ and $ f_2 $ can be obtained through Eq. (\ref{eq:faultreco}).
\end{theorem}
\begin{proof}
 Let the fault estimation error $ e_{f_i}=f_i(t)-\hat{f}_i(t) ,\quad i=1,2 $. It is evident that if the fault estimation error converge to zero, the precise fault reconstruction can be achieved. To this end, considering Eqs. (\ref{eq:lips}), (\ref{eq:faultrec}) and (\ref{eq:faultreco}), it follows that:
	\begin{equation}
	\begin{array}{l}
	\lim_{t\rightarrow\infty} \|e_{f_i}\| =
	\lim_{t\rightarrow\infty} \|f_i(t)-\hat{f}_i(t)\| =\\
	\lim_{t\rightarrow\infty} \|\phi_i^{-1}(y)\tilde{G}_i(0,e_{z_1})\| \leq \\
	C_d w\sqrt{\frac{2}{\rho}}\bar{\phi}_1 \lvert \sqrt{\Delta P_1} \rvert  \|\phi_i^{-1}(y)\| \lim_{t\rightarrow\infty}\|e_{z_1}\| = 0 ,\quad i=1,2
	\end{array}
	\end{equation}
	Theorem \ref{theorem2} is proven.
\end{proof}
\begin{remark}
			According to Eq. (11), terms $ \phi_1(y) = \frac{\beta}{V_{01}+A_1y_4}  $ , $ \phi_2(y) = \frac{\beta}{V_{02}-A_2y_4}  $ are physically positive and bounded on the overall region $ y \in \boldsymbol{{Y}} $, in which $ \boldsymbol{{Y}} $ is a compact set representing the system working area. Therefore, in domain $ \boldsymbol{{Y}} $, terms $ \phi_i(y), i=1,2  $ are nonsingular and their invertibility does not impose any strong condition. 
\end{remark}
\section{Experimental results and discussion}
\label{s:5}
Experimental tests were performed to assess the performance of the proposed ASTW observer for the objective of the fault reconstruction in electro-hydraulic servo systems.

\subsection{Setup description}\label{s5-1}

\begin{figure}[t]
	\centering
	\includegraphics[width=0.47\textwidth]{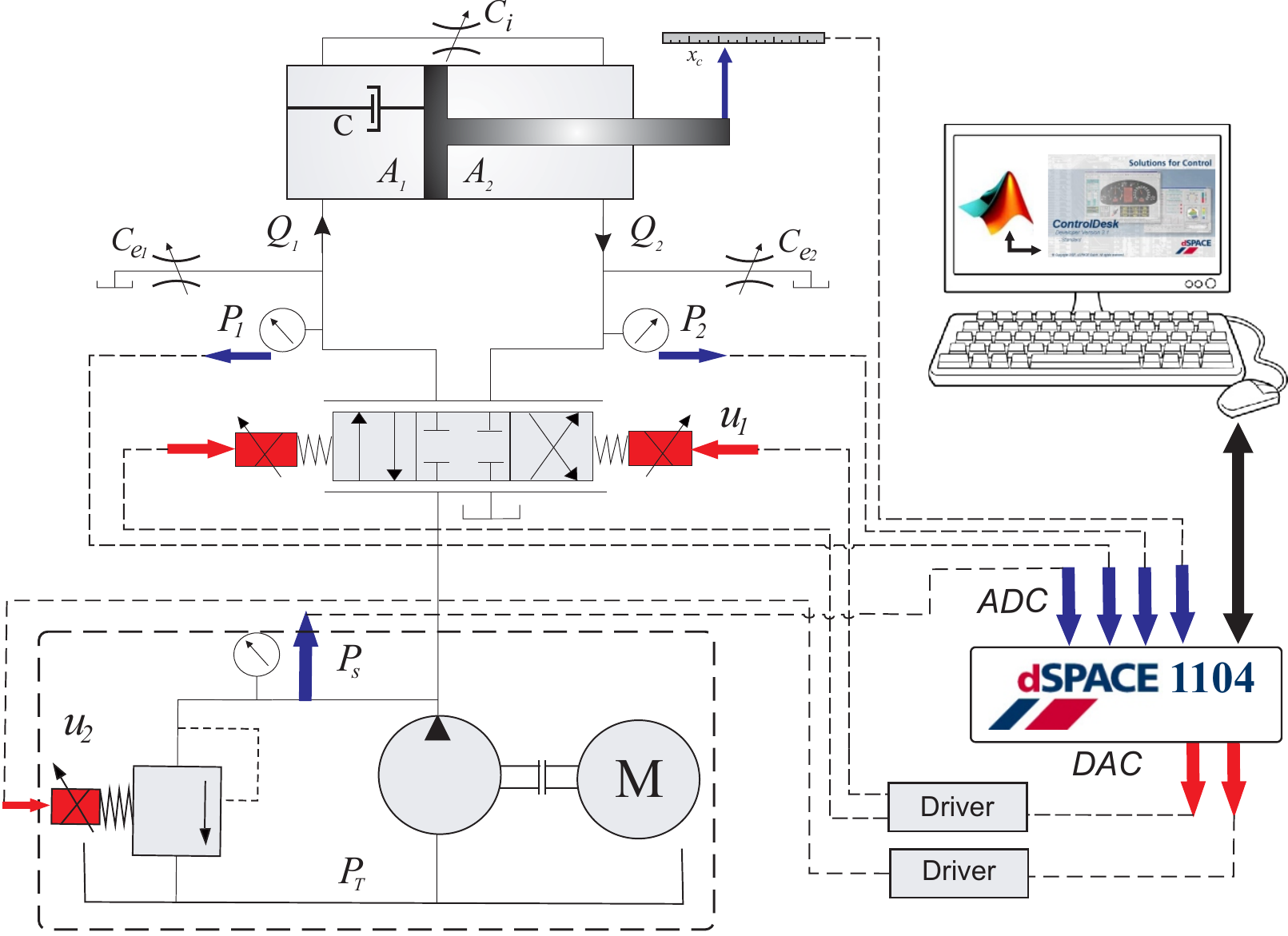}
	\caption{Schematic diagram of the hydraulic system.}\label{EHSSschema}
\end{figure}
%%%
%%%
\begin{figure}[t]
	\centering
	\includegraphics[width=0.47\textwidth]{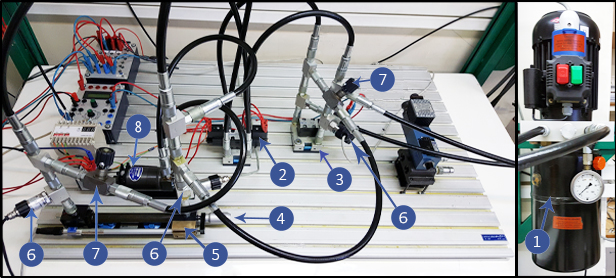}
	\caption{Picture of the experimental setup: 1 power unit, 2 proportional directional control valve, 3 proportional relief valve, 4 double-acting actuator, 5 position transducer, 6 pressure transducers, 7 ball valves, 8 flowmeter.}\label{testbench}
\end{figure}

The experimental set-up (see Figs. \ref{EHSSschema}-\ref{testbench}) is composed of a fixed-displacement pump unit \{1\} producing a constant flow $ Q_p = 2 Lit/min $, a proportional directional valve manufactured by FESTO (4/3 way-PDV 167086) \{2\} and a proportional relief valve manufactured by FESTO (PRV-FESTO 167087) \{3\} for flow and pressure control, respectively. A double acting cylinder \{4\} with a $ 0.2 m $ stroke measured by position transducer \{5\} and pressure transducers manufactured by HYDAC (HDA 4400) \{6\}. The supply pressure can be regulated up to $ P_s = 50bar $. The PDV receives control signal varying in the range of $u_1= -10 $ to $ +10 $ V and The PRV receives control signal varying in the range of $u_1= 0 $ to $ 5 $ V. As illustrated in Fig. \ref{EHSSschema}, system faults consist of both internal and external leakages. The internal leakage is produced artificially by bypassing the fluid across cylinder chambers, the external leakage is also produced by bypassing the fluid from each side of the cylinder to the reservoir. The severity of mentioned leakages is controlled via the adjustment of ball valves \{7\}. The flow rate of the internal leakage is measured by a positive-displacement flowmeter \{8\}. Finally, the nominal parameters of experimental set-up is presented in Table \ref{nominal parameters}.
 The observer algorithm is developed under \emph{Matlab-Simulink} environment and a dSPACE \emph{DS1104 R\&D Controller Board} with \emph{ControlDesk} user-interface is used for data acquisition and the real-time implementation of the proposed strategy. Control signals are sent by \emph{DAC} to the solenoid of EHSS valves and system outputs are received by \emph{ADC} from position and pressure transducers.\\
The sampling interval is set to $ T=1ms $.
\setlength{\tabcolsep}{6pt}
\begin{table}[h]
	\centering
	\caption{Nominal Parameters of experimental set-up}
	\label{nominal parameters}
	\begin{tabular}{ l l l}
		\hline
		Parameter & Description & Value\\
		\hline
		$ \beta $   & Effective bulk modulus & $ 1.05  $ GPa \\
		$ d_1 $ & Piston-side diameter & $ 16 $ mm \\
		$ d_2 $ & Rod-side diameter & $ 10 $ mm \\
		$ l $ & Piston stroke & $ 20 $ cm \\
		$ m $ & Equivalent mass & $ 0.15 $ Kg\\
		$ c $ & Equivalent damping & $ 350 $ N.s/m \\
		$ \rho $ & Fluid density & $ 845 $ Kg/$ \text{m}^3 $ \\
		$ C_d $ & PDV discharge coeff. & $ 0.7 $ \\
		$ K_v $ & PDV input gain & $ 1.13 \times 10^{-4} $ m/V \\
		$ \tau_v $ & PDV time constant & $ 0.07 $ sec \\
		$ K_r $ & PRV input gain & $ 1\times10^6 $ Pa/V \\
		$ \tau_s $ & PRV time constant & $ 0.05 $ sec \\
		\hline
	\end{tabular}
\end{table}
\subsection{Experimental results}\label{s5-2}

 The proposed strategy is evaluated in a closed-loop system in the presence of bounded leakage faults. In closed-loop tracking, a sinusoidal wave is generated as the desired position signal and the desired supply pressure is set to $ 30 bar $. The control inputs $ u_1 $ and $ u_2 $ are generated by a robust $ \mathscr{H_\infty} $ controller presented in our previous work \cite{bahrami2016energy}. In order to fulfill state estimation and fault reconstruction, the ASTW observer (\ref{eq:stwObs}) with gain adaptation law (\ref{eq:adapgain}) is implemented in the real-time environment of \emph{ControlDesk} and \emph{Matlab-Simulink}.

\begin{figure}[t]
	\centering
	\includegraphics[width=0.47\textwidth]{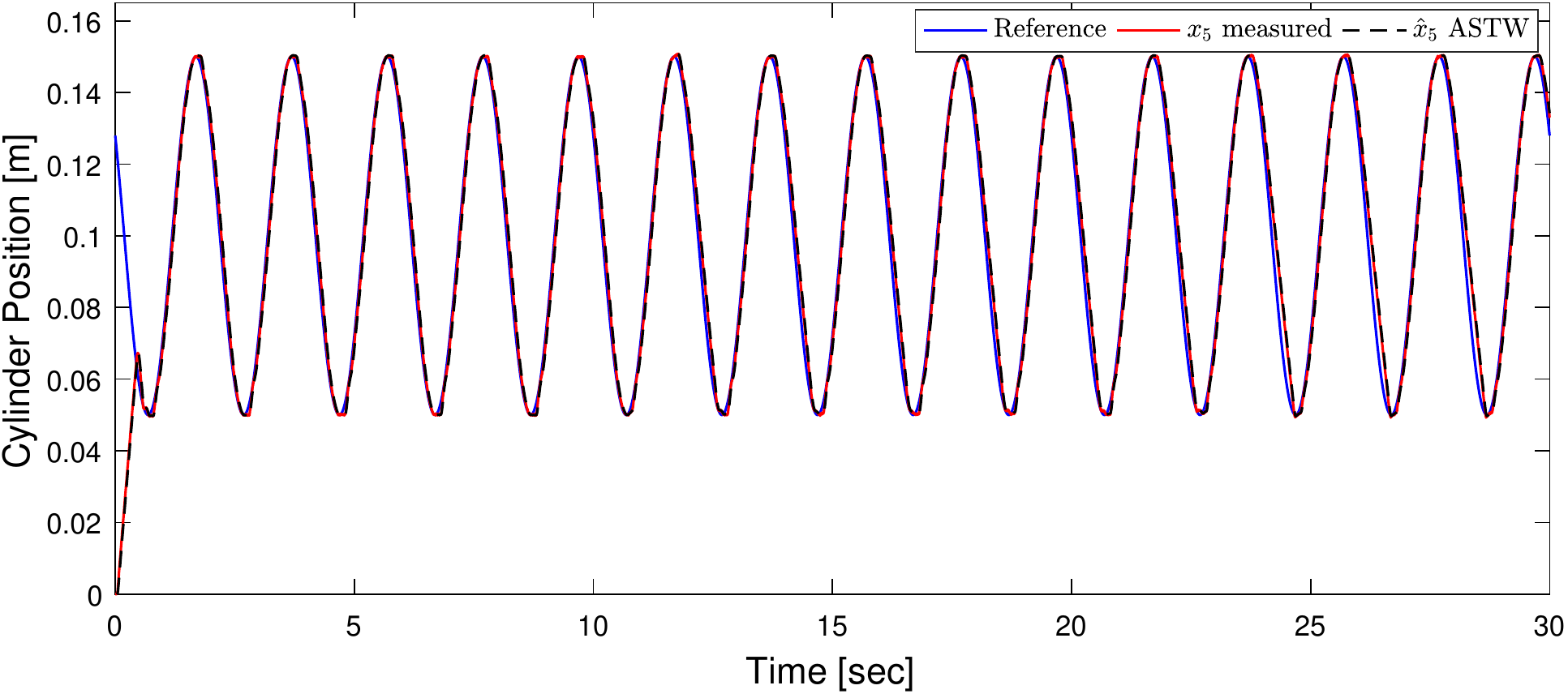}
	\caption{Actuator position and its estimation}\label{Pos}
\end{figure}
\begin{figure}[t]
	\centering
	\includegraphics[width=0.47\textwidth]{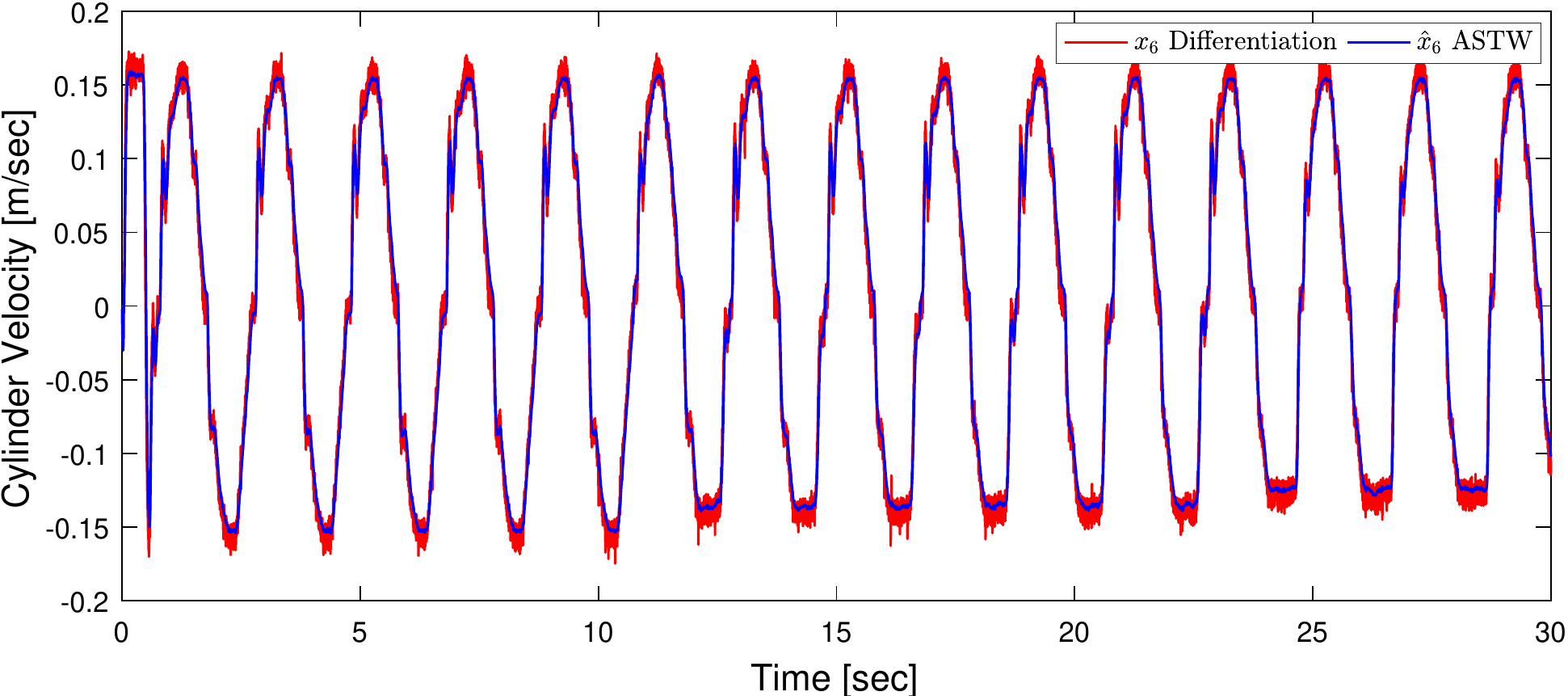}
	\caption{Actuator velocity and its estimation}\label{Vel}
\end{figure}
\begin{figure}[t]
	\centering
	\includegraphics[width=0.47\textwidth]{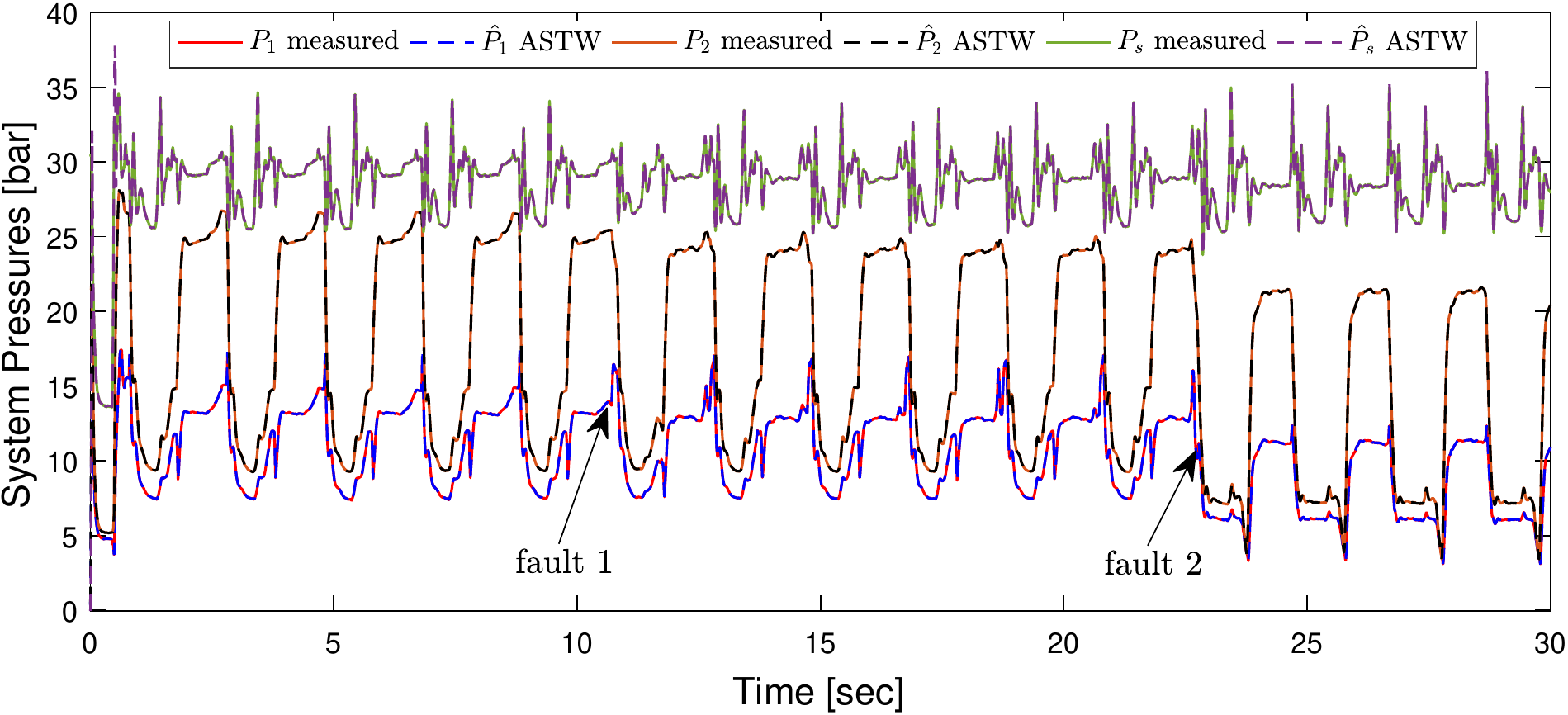}
	\caption{System pressures and their estimations with respect to the desired position}\label{Pressures}
\end{figure}
\begin{figure}[t]
	\centering
	\includegraphics[width=0.47\textwidth]{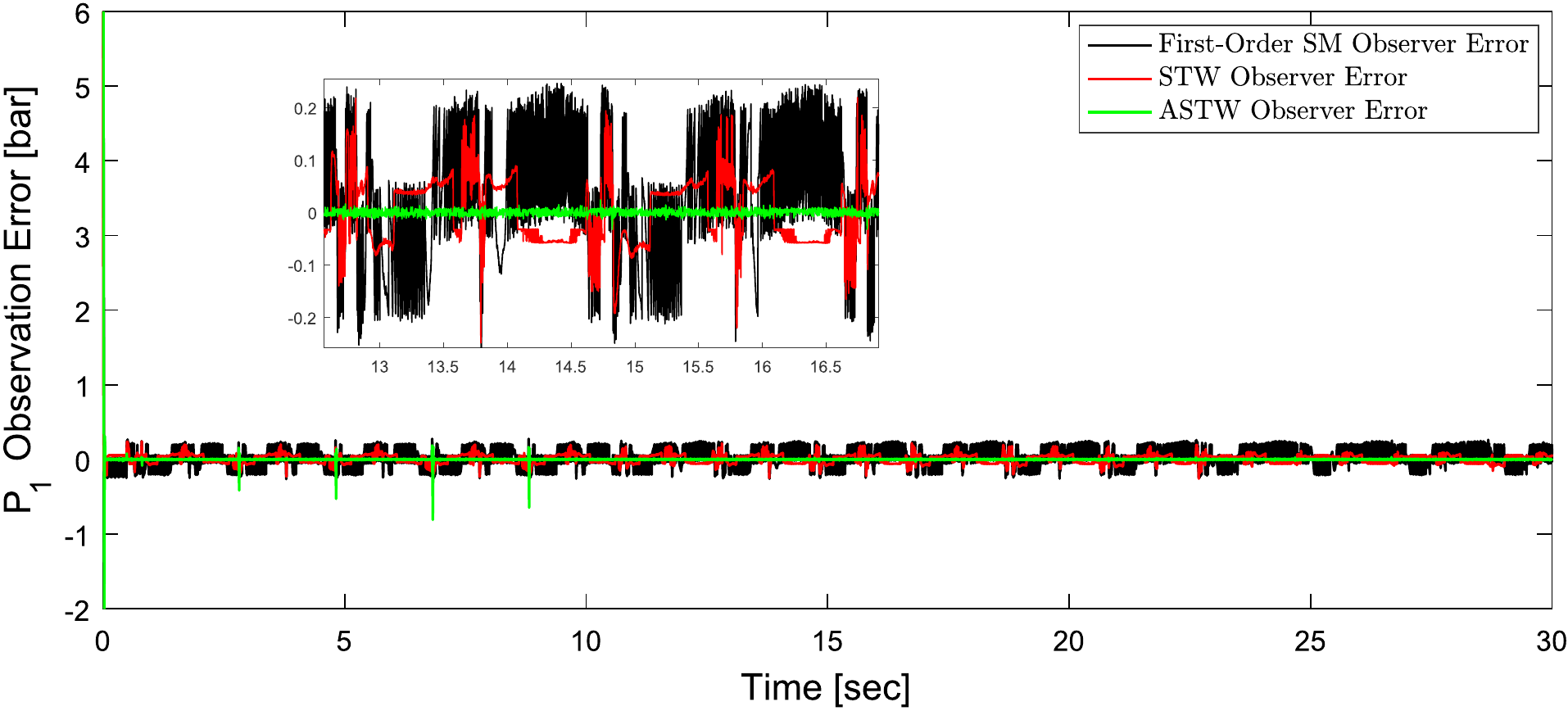}
	\caption{Estimation error of the piston-side pressure with different sliding mode techniques}\label{P1Error}
\end{figure}
\begin{figure}[t]
	\centering
	\includegraphics[width=0.47\textwidth]{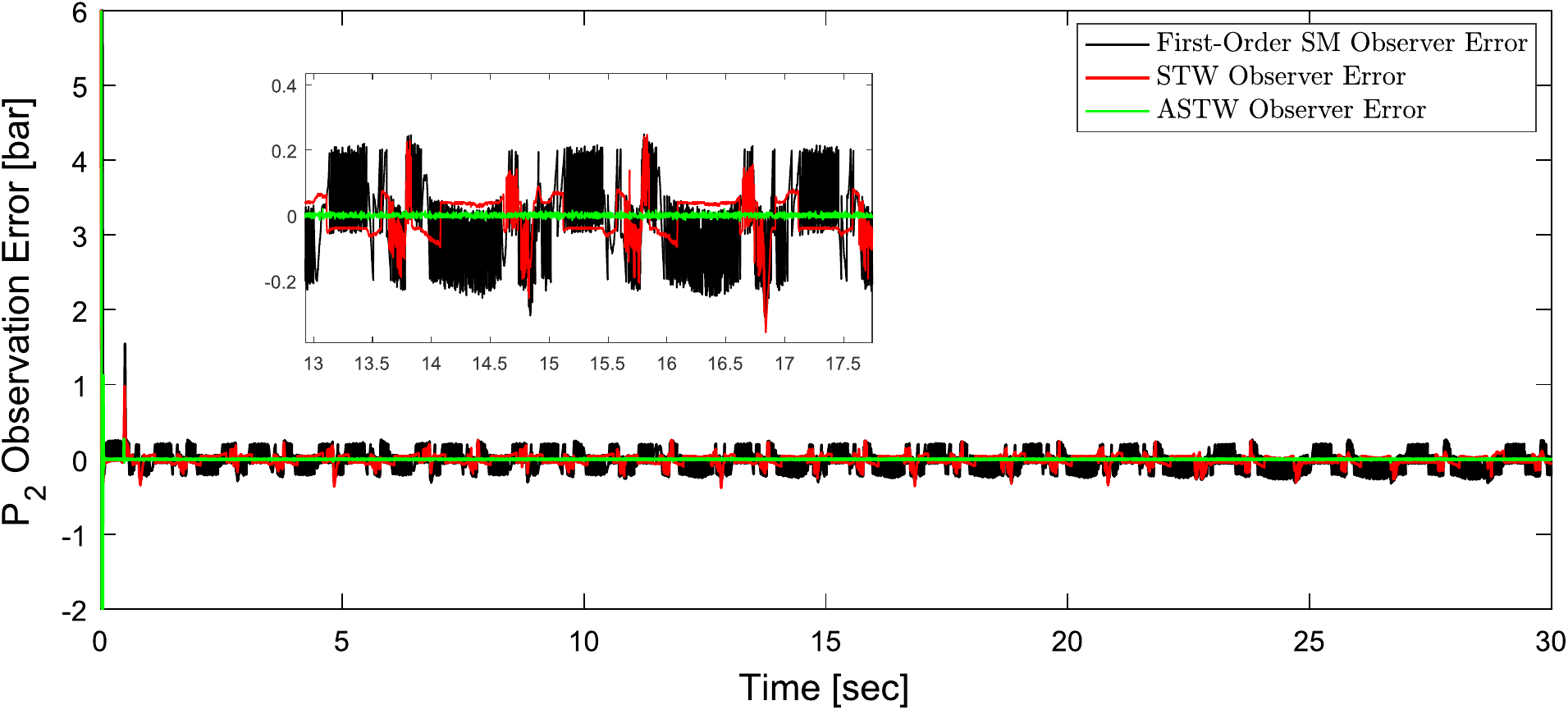}
	\caption{Estimation error of the rod-side pressure with different sliding mode techniques}\label{P2Error}
\end{figure}
\begin{figure}[h]
	\centering
	\includegraphics[width=0.47\textwidth]{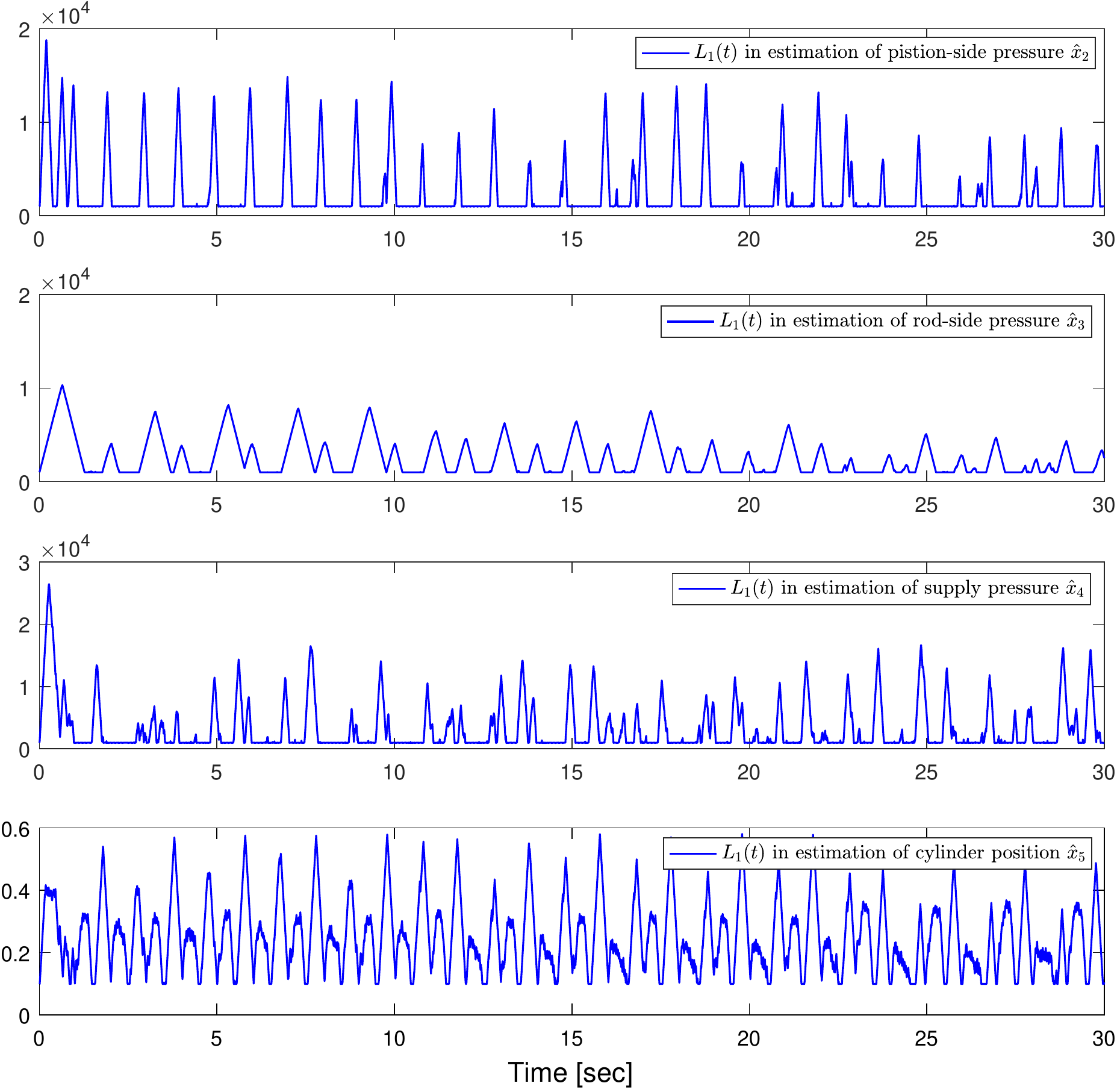}
	\caption{Adaptive gain $ L_1(t) $ in state estimation through ASTW observer }\label{adapgains}
\end{figure}
\begin{figure}[t]
	\centering
	\includegraphics[width=0.47\textwidth]{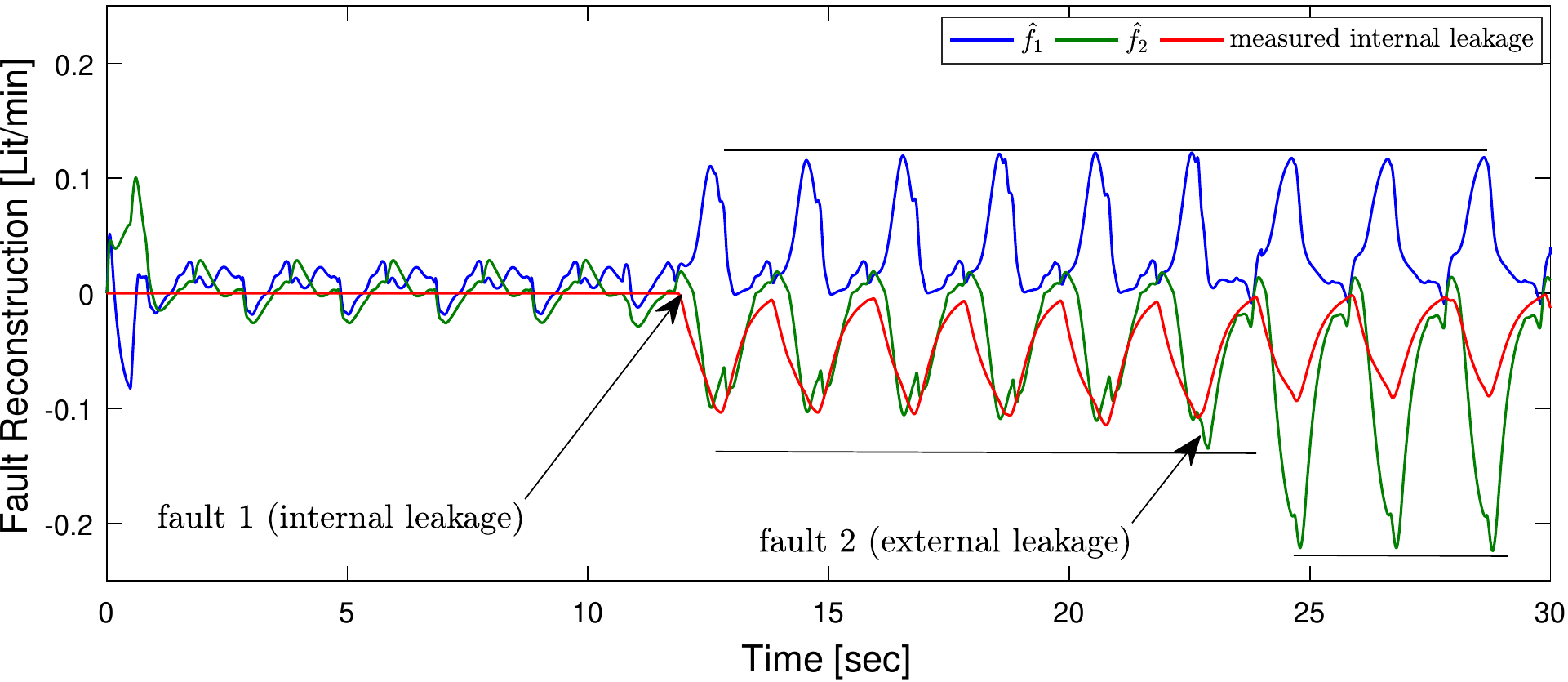}
	\caption{Leakage reconstruction}\label{mu}
\end{figure}
\begin{figure}[t]
	\centering
	\includegraphics[width=0.47\textwidth]{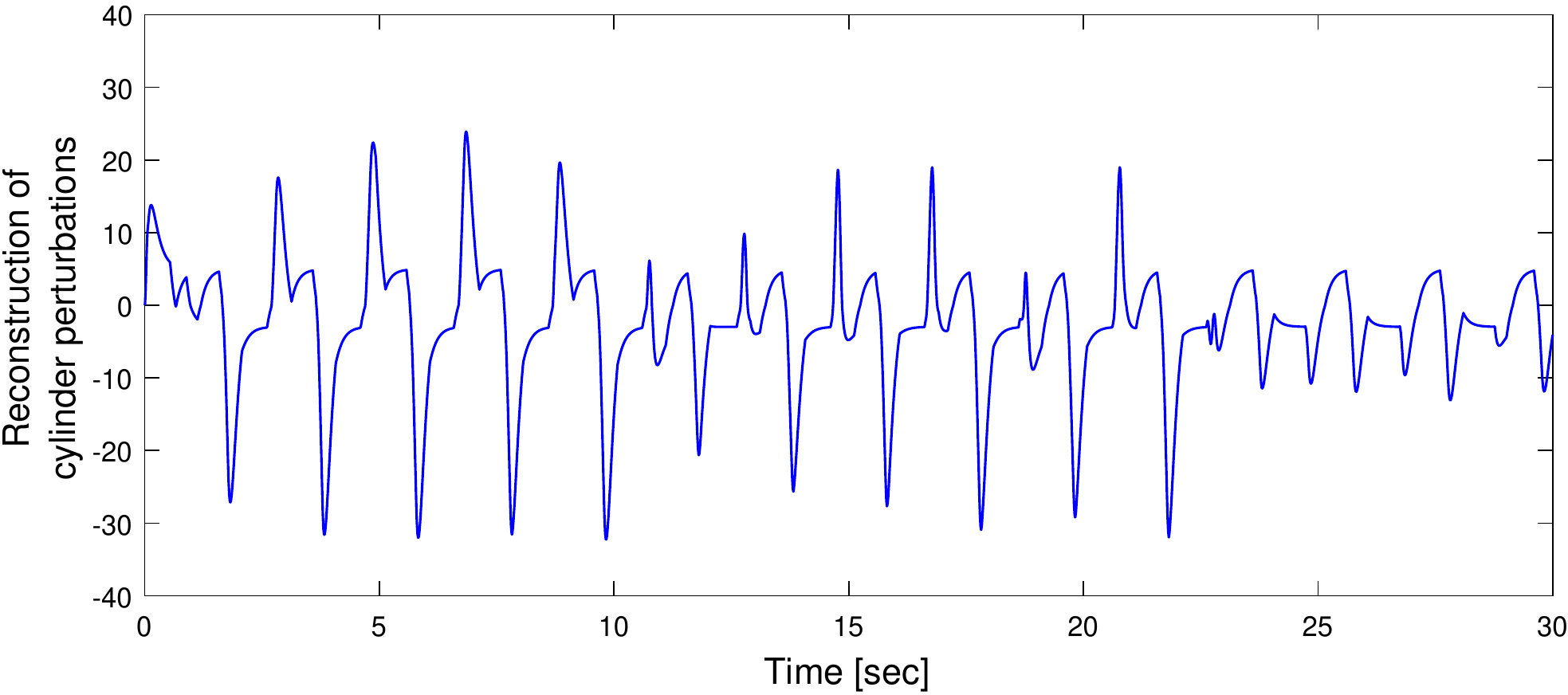}
	\caption{reconstruction of perturbations on cylinder including parametric uncertainties and friction}\label{mud}
\end{figure}

As mentioned earlier, two internal and external leakage faults were applied to the system through ball valves at different time instances (approximately at $ T_1\approx 12s $ and $ T_2\approx 23s $).\\
Figs. \ref{Pos}-\ref{Vel} present the cylinder position and its velocity, and their estimations. Clearly, the sliding motion is established in finite time and the ASTW observer has a good performance during the sliding motion in the presence of system faults. In Fig. \ref{Pressures}, system pressures and the observation of them are depicted. It is obvious that leakage faults make a pressure drop in the system, which can interrupt the actuator motion and closed-loop tracking. With this in mind, see Figs. \ref{Pos}-\ref{Vel} again. The effect of applied system faults is evident in the position and velocity of the backward movement of the cylinder. This is due to the fact that the cylinder moves based on the created pressure difference on it. Therefore, any faults (like leakage) affecting the system pressure also affects the cylinder movement, and conversely. 

In order to compare the effectiveness of the proposed ASTW algorithm with previous works \cite{edwards2000sliding,yan2007nonlinear,yan2008robust}, the observation error of the piston-side and rod-side pressure are obtained with First-Order Sliding Mode Observer, classic Super-twisting (STW) Observer and ASTW Observer. The results are presented in Figs. \ref{P1Error}-\ref{P2Error} . It can be observed that both chattering effect and observation error are highly reduced with the proposed ASTW observer. A more detailed comparison of observation errors for piston-and rod-side pressures are given in Tables \ref{Comparison P1} and \ref{Comparison P2}. In presented tables, performance indices are defined as follows: $ \lVert e_{y_i} \rVert_1 = \int_0^{30} |e_{y_i}(t)| dt $, $ \lVert e_{y_i} \rVert_2 = \sqrt{\int_0^{30} |e_{y_i}(t)| dt} $ and $ \lVert e_{y_i} \rVert_{\infty} = \sup_{t \in 10-30} |e_{y_i}(t)| $ , $ i=1,2.$ All of the performance indices are confirmed the effectiveness of the proposed methodology in presence of system nonidealities\footnote{In this case, due to the elimination of the effect of initial conditions, the infinity norm is defined on domain $ t = 10-30 seconds$ corresponding to the time interval of faults occurrence.}. This is due to the fact that observer gains increase and decrease in accordance with the effect of uncertainties/faults on the deviation of the sliding variable from the sliding surface (\ref{eq:slidingsurf}). For better interpretation, the adaptation of ASTW-observer gains is illustrated in Fig \ref{adapgains}. The gain adaptation enlarges the \emph{region of attraction} $ (\bar{\epsilon}) $ in the presence of perturbations and shrinks it in the absence of perturbations. Consequently, the conservatism of observer gains (gain overestimation) is eliminated and robust state estimation without chattering is obtained. These achievements are guaranteed during the sliding motion on the sliding surface (\ref{eq:slidingsurf}).

\setlength{\tabcolsep}{6pt}
\begin{table*}[t]
	\centering
	\caption{Comparison of $ P_1 $ observation errors}
	\label{Comparison P1}
	\begin{tabular}{ r c c c}
		\hline
		\textbf{Methodology} & $ \boldsymbol{\lVert e_{y_1}(t) \rVert_1} $ & $ \boldsymbol{\lVert e_{y_1}(t) \rVert}_2 $ & $ \boldsymbol{\lVert e_{y_1}(t) \rVert}_{\infty} $\\
		\hline
		First-Order SM Obs.   & $2.84 $ & $ 0.98 $ & $ 0.28 $ \\
		Super-twisting Obs. (STW) & $ 1.62 $ & $ 0.48 $ & $ 0.26 $ \\
		Adaptive Super-twisting Obs. (ASTW) & $ 0.25 $ & $ 0.83 $ & $ 0.03 $\\
		\hline
	\end{tabular}
\end{table*}
\setlength{\tabcolsep}{6pt}
\begin{table*}[t]
	\centering
	\caption{Comparison of $ P_2 $ observation errors}
	\label{Comparison P2}
	\begin{tabular}{ r c c c}
		\hline
		\textbf{Methodology} & $ \boldsymbol{\lVert e_{y_2}(t) \rVert_1} $ & $ \boldsymbol{\lVert e_{y_2}(t) \rVert}_2 $ & $ \boldsymbol{\lVert e_{y_2}(t) \rVert}_{\infty} $\\
		\hline
		First-Order SM Obs.   & $3.55 $ & $ 3.25 $ & $ 0.32 $ \\
		Super-twisting Obs. (STW) & $ 1.75 $ & $ 1.56 $ & $ 0.38 $ \\
		Adaptive Super-twisting Obs. (ASTW) & $ 0.57 $ & $ 1.38 $ & $ 0.04 $\\
		\hline
	\end{tabular}
\end{table*}

The estimation of leakage faults is shown in Fig. \ref{mu}. Two levels of the internal leakage \{$ C_{i}(P_2-P_1) $\} and external leakage \{$ C_{e2}(P_2-P_T) $\} are applied at the instants $ T\approx12s $ and $ T\approx23s $, respectively. The effect of applied leakage faults is evident in the amount of the equivalent output-error-injection signals ($ \mu(\cdot) $). Both signals $ \hat{f}_1 $ and $ \hat{f}_2 $ represent the amount of actuator leakages, which can be used in a fault-tolerant control strategy. More precisely, in this test, signal $ \hat{f}_1 $ denotes the amount of internal leakage \{$ C_{i}(P_2-P_1) $\} (leakage flow from the rod-side to the piston-side of the cylinder) after the instant $ T\approx12s $. Equally, signal $ \hat{f}_2 $ denotes the same leakage with opposite direction between the instants $ T\approx12s$ to $ 23s $. After applying the external leakage in the return line of the cylinder (instants $ T\approx23s$), signal $ \hat{f}_2 $ denotes the total amount of internal and external leakages \{$Q_{L2}=-C_{i}(P_2-P_1)-C_{e2}(P_2-P_T) $\} in the return line. The measured signal by flowmeter validates the reconstructed internal leakage.

Note that, the injection signals have a non-zero small value (fluctuating with zero-mean) before applying the leakages in the first 12 seconds. This is due to the hard nonlinear behavior of the EHSS such as the PDV dead-zone and other nonlinear uncertainties which are ignored in the system modeling and observer designing. Consequently, the fluctuating behavior of injection signals with zero-mean represents the bound of unmodeled dynamics and precise fault reconstruction can be achieved for any leakage faults bigger than this bound.
It is worth noting that the adverse effect of unmodeled dynamics on precise fault reconstruction can be eliminated through considering more knowledge about uncertainties and their bounds while designing the ASTW observer. It will be considered in a future study.

As mentioned in Section \ref{s:2}, using the proposed form (\ref{eq:nonlinearDyn2}) and ASTW observer (\ref{eq:stwObs}) provides the reconstruction of any matched and unmatched perturbations including uncertainties, disturbances, friction and faults that exist on cylinder dynamics. In this case, experimental tests are performed in a free-motion condition for the cylinder. Therefore, signal $ \mu_2(\sigma) $ in Fig. \ref{mud} represents the amount of matched perturbation including parametric uncertainties and the unmodeled friction of the cylinder (i.e. $ \varrho_4(e_{z_2},t) $ in Eq. \ref{eq:obsV}). 
\section{Conclusions}\label{s:6}
An adaptive-gain super-twisting observer is proposed for the aim of fault reconstruction in electro-hydraulic servo systems. The electro-hydraulic servo system receives bounded faults and uncertainties with unknown bounds. The proposed observer relaxes constraints on the distribution of system faults and uncertainties that needs to exist in such a way that fault reconstruction can be achieved. Furthermore, the proposed gain-adaptation algorithm takes into account system faults/uncertainties with unknown bounds leads to eliminate the overestimation of observer gains and drastically reduce the chattering effect. The precise leakage fault reconstruction was achieved. The performance of the proposed adaptive-gain super-twisting observer for fault reconstruction is confirmed by experimental tests on an electro-hydraulic servo system.
Future work will focus on extending the proposed methodology for simultaneous state and fault estimation to uncertain nonlinear systems.
%%
%%


%\section*{References}

\bibliography{ISATRA2729}

\begin{thebibliography}{10}
\expandafter\ifx\csname url\endcsname\relax
  \def\url#1{\texttt{#1}}\fi
\expandafter\ifx\csname urlprefix\endcsname\relax\def\urlprefix{URL }\fi
\expandafter\ifx\csname href\endcsname\relax
  \def\href#1#2{#2} \def\path#1{#1}\fi

\bibitem{merritt1967hydraulic}
H.~E. Merritt, Hydraulic control systems, John Wiley \& Sons, 1967 (1967).

\bibitem{rajamani1995adaptive}
R.~Rajamani, J.~K. Hedrick, Adaptive observers for active automotive
  suspensions: theory and experiment, IEEE Transactions on control systems
  technology 3~(1) (1995) 86--93 (1995).

\bibitem{alleyne2000simplified}
A.~Alleyne, R.~Liu, A simplified approach to force control for
  electro-hydraulic systems, Control Engineering Practice 8~(12) (2000)
  1347--1356 (2000).

\bibitem{baghestan2014robust}
K.~Baghestan, S.~Rezaei, H.~Talebi, M.~Zareinejad, Robust force control in a
  novel electro-hydraulic structure using polytopic uncertainty representation,
  ISA transactions 53~(6) (2014) 1873--1880 (2014).

\bibitem{yao2000adaptive}
B.~Yao, F.~Bu, J.~Reedy, G.-C. Chiu, Adaptive robust motion control of
  single-rod hydraulic actuators: theory and experiments, IEEE/ASME
  transactions on mechatronics 5~(1) (2000) 79--91 (2000).

\bibitem{baghestan2015energy}
K.~Baghestan, S.~M. Rezaei, H.~A. Talebi, M.~Zareinejad, An energy-saving
  nonlinear position control strategy for electro-hydraulic servo systems, ISA
  transactions 59 (2015) 268--279 (2015).

\bibitem{tivay2014switched}
A.~Tivay, M.~Zareinejad, S.~M. Rezaei, K.~Baghestan, A switched energy saving
  position controller for variable-pressure electro-hydraulic servo systems,
  ISA transactions 53~(4) (2014) 1297--1306 (2014).

\bibitem{chiang2005integrated}
M.~Chiang, Y.~Yeh, F.~Yang, Y.~Chen, Integrated control of clamping force and
  energy-saving in hydraulic injection moulding machines using decoupling fuzzy
  sliding-mode control, The International Journal of Advanced Manufacturing
  Technology 27~(1) (2005) 53--62 (2005).

\bibitem{chiang2004experimental}
M.-H. Chiang, C.-C. Huang, Experimental implementation of complex path tracking
  control for large robotic hydraulic excavators, The International Journal of
  Advanced Manufacturing Technology 23~(1) (2004) 126--132 (2004).

\bibitem{karpenko2009hardware}
M.~Karpenko, N.~Sepehri, Hardware-in-the-loop simulator for research on fault
  tolerant control of electrohydraulic actuators in a flight control
  application, Mechatronics 19~(7) (2009) 1067--1077 (2009).

\bibitem{gao2015survey}
Z.~Gao, C.~Cecati, S.~X. Ding, A survey of fault diagnosis and fault-tolerant
  techniques—part i: Fault diagnosis with model-based and signal-based
  approaches, IEEE Transactions on Industrial Electronics 62~(6) (2015)
  3757--3767 (2015).

\bibitem{sharifi2018leakage}
S.~Sharifi, A.~Tivay, S.~M. Rezaei, M.~Zareinejad, B.~Mollaei-Dariani, Leakage
  fault detection in electro-hydraulic servo systems using a nonlinear
  representation learning approach, ISA Transactions (2018).

\bibitem{goharrizi2010wavelet}
A.~Y. Goharrizi, N.~Sepehri, A wavelet-based approach to internal seal damage
  diagnosis in hydraulic actuators, IEEE transactions on industrial electronics
  57~(5) (2010) 1755--1763 (2010).

\bibitem{goharrizi2011wavelet}
A.~Y. Goharrizi, N.~Sepehri, A wavelet-based approach for external leakage
  detection and isolation from internal leakage in valve-controlled hydraulic
  actuators, IEEE Transactions on Industrial Electronics 58~(9) (2011)
  4374--4384 (2011).

\bibitem{goharrizi2012internal}
A.~Y. Goharrizi, N.~Sepehri, Internal leakage detection in hydraulic actuators
  using empirical mode decomposition and hilbert spectrum, IEEE Transactions on
  Instrumentation and Measurement 61~(2) (2012) 368--378 (2012).

\bibitem{shi2005development}
Z.~Shi, F.~Gu, B.~Lennox, A.~Ball, The development of an adaptive threshold for
  model-based fault detection of a nonlinear electro-hydraulic system, Control
  Engineering Practice 13~(11) (2005) 1357--1367 (2005).

\bibitem{chadli2009state}
M.~Chadli, A.~Akhenak, J.~Ragot, D.~Maquin, State and unknown input estimation
  for discrete time multiple model, Journal of the Franklin Institute 346~(6)
  (2009) 593--610 (2009).

\bibitem{an2005hydraulic}
L.~An, N.~Sepehri, Hydraulic actuator leakage fault detection using extended
  kalman filter, International Journal of Fluid Power 6~(1) (2005) 41--51
  (2005).

\bibitem{hammouri2002failure}
H.~Hammouri, P.~Kabore, S.~Othman, J.~Biston, Failure diagnosis and nonlinear
  observer. application to a hydraulic process, Journal of the Franklin
  Institute 339~(4) (2002) 455--478 (2002).

\bibitem{mahulkar2011derivative}
V.~Mahulkar, D.~E. Adams, M.~Derriso, Derivative free filtering in hydraulic
  systems for fault identification, Control Engineering Practice 19~(7) (2011)
  649--657 (2011).

\bibitem{edwards2000sliding}
C.~Edwards, S.~K. Spurgeon, R.~J. Patton, Sliding mode observers for fault
  detection and isolation, Automatica 36~(4) (2000) 541--553 (2000).

\bibitem{yan2007nonlinear}
X.-G. Yan, C.~Edwards, Nonlinear robust fault reconstruction and estimation
  using a sliding mode observer, Automatica 43~(9) (2007) 1605--1614 (2007).

\bibitem{yan2008robust}
X.-G. Yan, C.~Edwards, Robust sliding mode observer-based actuator fault
  detection and isolation for a class of nonlinear systems, International
  Journal of Systems Science 39~(4) (2008) 349--359 (2008).

\bibitem{levant1993sliding}
A.~Levant, Sliding order and sliding accuracy in sliding mode control,
  International journal of control 58~(6) (1993) 1247--1263 (1993).

\bibitem{moreno2008lyapunov}
J.~A. Moreno, M.~Osorio, A lyapunov approach to second-order sliding mode
  controllers and observers, in: Decision and Control, 2008. CDC 2008. 47th
  IEEE Conference on, IEEE, 2008, pp. 2856--2861 (2008).

\bibitem{shtessel2012novel}
Y.~Shtessel, M.~Taleb, F.~Plestan, A novel adaptive-gain supertwisting sliding
  mode controller: Methodology and application, Automatica 48~(5) (2012)
  759--769 (2012).

\bibitem{bahrami2016energy}
M.~Bahrami, A.~Tivay, K.~Baghestan, S.~Rezaei, M.~Zareinejad, An energy-saving
  robust motion control of redundant electro-hydraulic servo systems, in:
  Robotics and Mechatronics (ICROM), 2016 4th International Conference on,
  IEEE, 2016, pp. 459--464 (2016).

\bibitem{tivay2013energy}
A.~Tivay, S.~M. Rezaei, K.~Baghestan, M.~Zareinejad, Energy-saving cooperative
  position tracking control of electro-hydraulic servo systems, in: Robotics
  and Mechatronics (ICRoM), 2013 First RSI/ISM International Conference on,
  IEEE, 2013, pp. 511--516 (2013).

\bibitem{perruquetti2002sliding}
W.~Perruquetti, J.-P. Barbot, Sliding mode control in engineering, CRC Press,
  2002 (2002).

\bibitem{khalil1996noninear}
H.~K. Khalil, Nonlinear systems, Prentice-Hall, New Jersey 2~(5) (1996) 5--1
  (1996).

\end{thebibliography}

\end{document}